\documentclass[11pt,a4paper]{amsart}
\usepackage{amssymb,dsfont}
\usepackage[left=25mm,top=30mm,bottom=25mm,right=25mm]{geometry}
\usepackage{cmap}                                                   % make the PDF file "searchable and copyable"
\usepackage{hyperxmp}
\usepackage[utf8]{inputenc}
\usepackage[english]{babel}
\usepackage[fixlanguage]{babelbib}
\usepackage{graphicx}
\usepackage[pdfdisplaydoctitle=true,
      colorlinks=true,
      urlcolor=blue,
      citecolor=blue,
      linkcolor=blue,
      pdfstartview=FitH,
      pdfpagemode=UseNone,
      bookmarksnumbered=true,
      unicode=true]{hyperref}                                       % to make a pdf file with hyperlinks and bookmarks
\usepackage{bm}
\usepackage{calrsfs}
\usepackage[all]{xy}
\usepackage{tensor}

\newtheorem{thm}{Theorem}[section]

\newtheorem{lem}[thm]{Lemma}

\theoremstyle{definition}

\newtheorem{exam}[thm]{Example}

\theoremstyle{remark}
\newtheorem{rem}[thm]{Remark}

\numberwithin{equation}{section}

\newcommand{\GL}{\mathrm{GL}}               % General linear group
\newcommand{\RR}{\mathbb{R}}                % Real numbers
\newcommand{\Sym}{\mathbb{S}}               % Space of totally symmetric covariant tensors
\newcommand{\espace}{\mathcal{E}}           % L'espace affine euclidien
\newcommand{\Vect}{\mathrm{Vect}}

\newcommand{\bgamma}{{\bm{\gamma}}}
\newcommand{\bsigma}{{\bm{\sigma}}}
\newcommand{\btheta}{{\bm{\theta}}}
\newcommand{\btau}{{\bm{\tau}}}
\newcommand{\bbeta}{\bm{\beta}}

\newcommand{\bSigma}{{\bm{\Sigma}}}
\newcommand{\bkappa}{\bm{\kappa}}

\newcommand{\norm}[1]{\left\Vert#1\right\Vert}

\newcommand{\set}[1]{\left\{#1\right\}}

\newcommand{\bq}{q}                         % Euclidean metric
\newcommand{\got}{\mathfrak{g}}             % Galilean co-metric
\newcommand{\vol}{\mathrm{vol}}             % Riemannian volume form
\newcommand{\Ric}{\mathbf{Ric}}             % Ricci tensor
\newcommand{\Ein}{\mathbf{G}}               % Einstein tensor
\newcommand{\bt}{\mathbf{t}}
\newcommand{\VV}{\bm{V}}                    % Vitesse Lagrangienne
\newcommand{\aaa}{\bm{a}}
\newcommand{\bb}{\mathbf{b}}                % tenseur d'ordre 2
\newcommand{\bd}{\mathbf{d}}
\newcommand{\be}{\mathbf{e}}
\newcommand{\bk}{\mathbf{k}}                % tenseur covariant
\newcommand{\bs}{\mathbf{s}}                % microstress

\newcommand{\bS}{\mathbf{S}}
\newcommand{\bT}{\mathbf{T}}
\newcommand{\vv}{\bm{v}}

\newcommand{\xx}{\mathbf{x}}
\newcommand{\yy}{\mathbf{y}}

\newcommand{\uu}{\bm{u}}
\newcommand{\pp}{p}
\newcommand{\bp}{\bm{p}}
\newcommand{\ee}{\bm{e}}
\newcommand{\bA}{\mathbf{A}}

\newcommand{\bC}{\mathbf{C}}

\newcommand{\bF}{\mathbf{F}}

\newcommand{\bH}{\mathbf{H}}

\newcommand{\mH}{\mathcal{H}}

\newcommand{\bK}{\mathbf{K}}

\newcommand{\mL}{\mathcal{L}}
\newcommand{\mW}{\mathcal{W}}

\newcommand{\bN}{\mathbf{N}}
\newcommand{\bP}{\mathbf{P}}

\newcommand{\bU}{\mathbf{U}}

\newcommand{\bW}{\mathbf{W}}
\newcommand{\bX}{\mathbf{X}}

\newcommand{\mM}{\mathcal{M}}               % Univers
\newcommand{\mg}{{g^{3D}}}
\newcommand{\body}{\mathcal{B}}

\newcommand{\id}{\mathrm{id}}

\newcommand{\dd}{\mathrm{d}}

\newcommand{\rhoz}{{\bar{\rho}}}
\newcommand{\uuz}{{\bar{\uu}}}
\newcommand{\uz}{{\bar{u}}}
\newcommand{\rhol}{{\overset{\scriptscriptstyle\lambda}{\rho}}}
\newcommand{\uul}{{\overset{\scriptscriptstyle\lambda}{\uu}}{}}
\newcommand{\gl}{{\overset{\scriptscriptstyle\lambda}{g}}{}}

\newcommand{\Nl}{{\mathcal{N}}}
\newcommand{\kl}{k}
\newcommand{\Pl}{{\overset{\scriptscriptstyle\lambda}{\bP}}}

\newcommand{\Hl}{{\overset{\scriptscriptstyle\lambda}{\bH}}}
\newcommand{\Hz}{{\bar{\bH}}}
\newcommand{\Tl}{{\overset{\scriptscriptstyle\lambda}{\bT}}}
\newcommand{\Tz}{{\bar{\bT}}}
\newcommand{\El}{{\overset{\scriptscriptstyle\lambda}{E}}}
\newcommand{\Ez}{{\bar{E}}}
\newcommand{\pl}{{\overset{\scriptscriptstyle\lambda}{\bp}}{}}
\newcommand{\bsl}{{\overset{\scriptscriptstyle\lambda}{\bs}}}

\newcommand{\bsigz}{{\bar{\bsigma}}}
\newcommand{\sigz}{{\bar{\sigma}}}

\DeclareMathOperator{\Lie}{L} %
\DeclareMathOperator{\tr}{tr} %
\DeclareMathOperator{\grad}{grad} %
\DeclareMathOperator{\dive}{div} %

\DeclareMathOperator{\divl}{\overset{\scriptscriptstyle\lambda}{\dive}} %
\DeclareMathOperator{\divz}{{\dive^{\text{\tiny NC}}}} %

\hypersetup{
    pdftitle={Souriau's Relativistic general covariant formulation of hyperelasticity revisited},
    pdfauthor={Boris Kolev and Rodrigue Desmorat},
    pdfsubject={MSC 2020: 74B20, 70G45, 83C10, 83C25},
    pdfkeywords={Constitutive equations, Relativistic Hyperelasticity, Lagrangian formulation of General Relativity, Newton-Cartan theory of Continuum Mechanics},
    pdflang=fr
}
\begin{document}

\title[Souriau's relativistic formulation of hyperelasticity revisited]{Souriau's Relativistic general covariant \protect\\  formulation of hyperelasticity revisited}

\author{B. Kolev}
\address[Boris Kolev]{Université Paris-Saclay, ENS Paris-Saclay, CentraleSupélec, CNRS, LMPS - Laboratoire de Mécanique Paris-Saclay, 91190, Gif-sur-Yvette, France}
\email{boris.kolev@ens-paris-saclay.fr}

\author{R. Desmorat}
\address[Rodrigue Desmorat]{Université Paris-Saclay, ENS Paris-Saclay, CentraleSupélec, CNRS, LMPS - Laboratoire de Mécanique Paris-Saclay, 91190, Gif-sur-Yvette, France}
\email{rodrigue.desmorat@ens-paris-saclay.fr}

\date{\today}%

% ----------------------------------------------------------------
\begin{abstract}
  We present and modernize Souriau's 1958 geometric framework for Relativistic continuous media, and enlighten the necessary and the \emph{ad hoc} modeling choices made since, focusing as much as possible on the Continuum Mechanics point of view.
  We describe the general covariant formulation of Hyperelasticity in (Variational) General Relativity, and then in the particular case of a static spacetime. Different relativistic strain and stress tensors are formulated and discussed. Finally, we apply Souriau's formalism to Schwarzschild's metric, and recover the Classical Galilean Hyperelasticity with gravity, as the Newton--Cartan infinite light speed limit of this formulation.
\end{abstract}

\maketitle

% ----------------------------------------------------------------
\setcounter{tocdepth}{1}
\tableofcontents
% ----------------------------------------------------------------

% ----------------------------------------------------------------
\section*{Introduction}
% ----------------------------------------------------------------

Attempts to formulate \emph{Relativistic Elasticity} in the \emph{General Relativity} framework go back to 1916 with the pioneering work of Nordström~\cite{Nor1916}, in Dutch. Since then, several authors have first aimed at proposing constitutive equations for Relativistic fluids~\cite{Tau1954,Lic1955,CQ1972,MTW1973} and, then, at modeling Relativistic continuous media, most often at the astrophysics scale~\cite{Sou1958,Syn1959,deW1962,Ray1963,Sou1964,Ben1965,Old1970,Lam1989,KM1992,KM1997,BS2003,EBT2006,Wer2006,GHE2011,Bro2021}, for instance for the modeling of the solid crust of neutron stars, but also at a local scale~\cite{Mau1978a,Mau1978b,Mau1978c,PR2013,PR2013,PRA2015,NWP2022}, for mechanical engineering applications.

This has led Lichnerowicz to define \emph{pure matter}~\cite{Lic1955}, synonymous of dust, and Souriau to define \emph{perfect matter}~\cite{Sou1958,Sou1960,Sou1964}, as a continuous medium which can be described independently from electromagnetic phenomena. In the present work, we follow Souriau and model perfect matter with the \emph{Gauge Theory} mindset~\cite{Ble1981}. More precisely, we focus on Relativistic hyperelastic continuous media. We do not consider the coupling with electromagnetism, nor with temperature.

The work of Souriau (1958, in French), seems to have been unnoticed by the scientific community. It is prior to the works of Synge (1959), of DeWitt (1962) and of Rayner (1963) (all three criticized in the later papers by Bennoun~\cite{Ben1965} and Carter and Quintana~\cite{CQ1972}). As we shall see, Souriau did in fact formulate the correct framework to describe Relativistic Hyperelasticity, first in his long 1958 paper~\cite{Sou1958}, then in his 1964 book~\cite{Sou1964} (in French still). The modern geometric picture of the General Relativity framework for elastic media is, up to details that we shall discuss on the go, derived in~\cite{Sou1958,Sou1964}, and later in~\cite{CQ1972,KM1992,KM1997,BS2003}.

We stick to the chronology introduced by Souriau of the mathematical concepts, the idea being as often to make the least hypotheses as possible. This is why the introduction of a \emph{foliation by spacelike hypersurfaces} ---which is not assumed \emph{a priori}--- is only addressed in \autoref{sec:Universe-foliation}, and why the problems of the \emph{definition of time} and of the formulation of Relativistic Hyperelasticity in a \emph{spacetime} is addressed only in \autoref{sec:hyperelasticity-spacetime}.
We finally apply Künzle's methodology \cite{Kuen1976}, to mathematically recover Classical Galilean Hyperelasticity with gravity, as the Newton--Cartan infinite light speed limit of the described General Relativity formulation. Our calculations generalize the ones for relativistic fluids to relativistic solids.

We write this paper mainly with the mechanics ---not the astrophysics--- point of view. We seek for the adequacy with the geometric formulation on the body $\body$ of three-dimensional Hyperelasticity (the so-called \emph{intrinsic Lagrangian formulation}, developed by Noll~\cite{Nol1972,Nol1978} and Rougée~\cite{Rou1991a,Rou2006}, see also~\cite{KD2021,KD2021a}). The notations are chosen to be compatible with both the ones used classically in Continuum Mechanics of solid materials~\cite{LC1985} and the ones considered in~\cite{KD2021,KD2021a}.

\subsection*{Outline}

The article is organized as follows. In \autoref{sec:matter-fields}, we introduce the basic concepts of \emph{matter field} $\Psi$, of \emph{body World tube} $\mW$, of \emph{mass measure $\mu$}, and of \emph{current of matter} $\bP$, which are the starting point of the theory of Relativistic Continuum Mechanics. The normalization of the later allows to define the \emph{rest mass density} $\rho_{r}$ and the unit timelike vector $\bU$. The section ends with the definition of the \emph{conformation} $\bH$, the cornerstone of Souriau's general covariant formulation of Relativistic Hyperelasticity. Matter conservation is formulated in \autoref{sec:mass-conservation}. Definitions of relativistic strain tensors are provided in \autoref{sec:Conformation-strains}. Relativistic Hyperelasticity is formulated in \autoref{sec:Lagrangian-formulation}, in which the Lagrangian formulation of General Relativity is recalled and applied to this specific constitutive modeling. The \emph{stress-energy tensor} is introduced in \autoref{sec:stress-energy-tensor} and both four-dimensional and three-dimensional relativistic \emph{stress tensors} are defined. It is only in \autoref{sec:Universe-foliation} that a \emph{spacetime structure} is considered, allowing to better connect the preceding general geometric framework with Classical Continuum Mechanics, and to define the \emph{generalized Lorentz factor} $\gamma$ (\autoref{sec:Matter-field-spacetime}). This factor accounts for the distortion between the unit vector $\bU$ (the matter) and the unit normal $\bN$ to the spacelike hypersurfaces $\Omega_{t}$ (the observer), and allows for the geometric definition of the \emph{relativistic mass density} $\rho$. A first formulation of Relativistic Hyperelasticity in a static spacetime, including the generalization of the \emph{Cauchy stress tensor,} is derived in \autoref{sec:hyperelasticity-spacetime}. This framework is detailed in \autoref{sec:hyperelasticity-Schwarzschild} for the particular case of the \emph{Schwarzschild metric}. The Galilean (Newton--Cartan) infinite light speed limit of the theory is discussed in \autoref{sec:Gallilean}.

\subsection*{Notations}

Given a linear operator $L\colon E \to F$ between two finite dimensional vector spaces, we denote by $L^{\star} \colon F^{\star} \to E^{\star}$, $\beta \mapsto \beta \circ L$, its \emph{transpose}. If moreover, the vector space $E$ is equipped with an inner product $q_{E}$, and $F$, with an inner product $q_{F}$, we can define its \emph{adjoint}, defined implicitly by the relation $\overline{L} \colon F \to E$, $\langle Lv,w \rangle_{F} = \langle v,\overline{L}w \rangle_{E}$ for all $v \in E$ and $w \in F$. The relation between $L^{\star}$ and $\overline{L}$ is thus written as $\overline{L} = q_{E}^{-1} L^{\star} q_{F}$.
We denote by $\Sym^{k}E$ the set of totally symmetric tensors of order $k$ on $E$ and by $\Lambda^{k}E$ the set of alternate tensors of order $k$ on $E$.

Now, let $\mM$ be a differential manifold of dimension $n$, we denote by $\Omega^{k}(\mM)$, the set of \emph{differential $k$-forms} on $\mM$, that is smooth sections of the vector bundle $\Lambda^{k}T^{\star}\mM$. The contraction $i_{X} \alpha$, of components $(i_{X} \alpha)_{ \mu_{1} \dotsc \mu_{k-1}}=X^{\nu} \alpha_{\nu \mu_{1} \dotsc \mu_{k-1}}$, denotes the \emph{interior product} of a vector field $X\in \Vect(\mM)$ with a $k$-form $\alpha \in \Omega^{k}(\mM)$,
\begin{equation*}
  i_{X} \alpha:=\alpha(X, \cdot, \dotsc, \cdot) \in  \Omega^{k-1}(\mM).
\end{equation*}
If moreover, $\mM$ is endowed with a Riemannian or pseudo-Riemannian metric $g$ (and $\mM$ is orientable), we will denote by $\vol_{g}\in \Omega^{n}(\mM)$ the (pseudo-)Riemannian \emph{volume form} associated with $g$.

Given a 1-form $\alpha \in \Omega^{1}(\mM)$, the notation $\alpha^{\sharp}:=g^{-1}\alpha\in \Vect(\mM)$ stands for $\alpha^{\mu}=g^{\mu\nu} \alpha_{\nu}$. Conversely, given a vector field $X\in \Vect(\mM)$, $X^{\flat}:=gX\in \Omega^{1}(\mM)$ stands for $X_{\mu}=g_{\mu\nu} X^{\nu}$. When local coordinates are involved on a 4-dimensional manifold, the Greek subscripts or superscripts $\mu, \nu, \rho\dots$ range from 0 to 3, while the roman ones $i, j, k\dots$, or $I, J, K, \dots$ range from 1 to 3.

The light speed will be denoted by $c$ and we refer to the Galilean three-dimensional Continuum Mechanics of solids~\cite{TN1965,LC1985,MH1994} simply as \emph{Classical Continuum Mechanics}.

% ----------------------------------------------------------------
\section{Matter field, current of matter and conformation}
\label{sec:matter-fields}
% ----------------------------------------------------------------

The Universe is assumed to be a four-dimensional orientable manifold $\mM$, endowed with an hyperbolic metric $g$, of signature
$(-,+,+,+)$. Its pseudo-Riemannian volume form is denoted by
\begin{equation*}
  \vol_{g}\in \Omega^{4}(\mM).
\end{equation*}
In the present work we limit our study to a (non electromagnetic) continuous particles assembly, the so-called \emph{perfect matter}~\cite{Lic1955}. Its modeling adopted by Souriau in~\cite{Sou1958,Sou1960,Sou1964} is inspired by Gauge theory~\cite{Ble1981}, where \emph{matter fields} are described by sections of an \emph{associated bundle}, \textit{i.e.}, some vector bundle constructed using a linear representation of the \emph{structural group} of the considered Gauge theory on some given vector space. The specificity and relative simplicity of the present description of \emph{perfect matter} is, however, that we assume this linear representation, and thus the vector bundle, to be trivial. More precisely, we let $V$ be a three-dimensional vector space (taken as $\RR^{3}$ in~\cite{Sou1958}). A \emph{perfect matter field} (called the particles labelling in~\cite{Sou1958}, and the \emph{projection}, noted $\mathcal P$, in~\cite{CQ1972}) is then a smooth vector valued function
\begin{equation*}
  \Psi: \mM \to V.
\end{equation*}

\begin{rem}
  The notation $\Psi$ for the matter field is on purpose chosen similar to the one for the wave function in Quantum Mechanics.
\end{rem}

Matter is then described by the set of all the material points constitutive of the continuous medium under study in the Universe
(for example a mechanical structure). Their labels constitute a set $\body\subset V$, assumed to be (in general) a three-dimensional compact orientable manifold with boundary and called the \emph{body}. It is further assumed that $\Psi$ is a \emph{submersion} on $\mW=\Psi^{-1} (\body)$: the linear tangent map $T\Psi: T\mW \to TV$ is of rank 3 at each point of $\mW$. Thus, $\mW$ is fibered by the particles World lines $\Psi^{-1}(\bX)$, $\bX\in \body$, and is called for this reason the \emph{body's World tube}.

\begin{figure}[ht]\label{fig:fiberedtimelines}
  \centering
  \includegraphics[width=14cm]{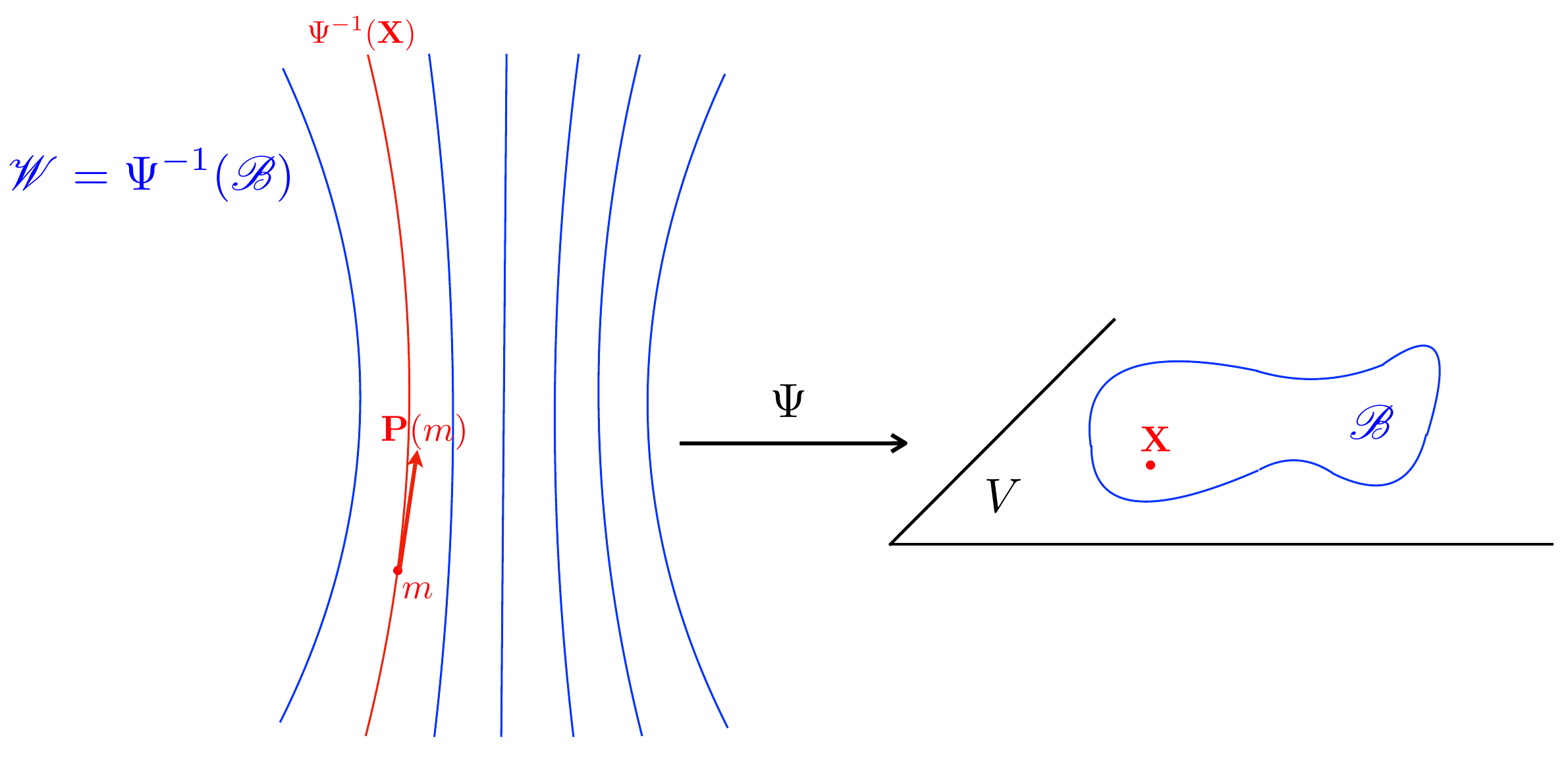}
  \caption{The World tube $\mW=\Psi^{-1} (\body)$ fibered by the particles World lines $\Psi^{-1}(\bX)$.}
\end{figure}

The body $\body$ is endowed with a volume form $\mu\in \Omega^{3}(\body)$, the \emph{mass measure}, which carries the information about the distribution of matter present in $\mW$~\cite{KM1992}. This interpretation is connected with the three-dimensional Classical Continuum Mechanics theory, in which the abstract manifold $\body$, equipped with the mass measure $\mu$, is in fact the \emph{body} introduced by Truesdell and Noll~\cite{TN1965,Nol1972,Nol1974,Nol1978}.

As we seek for a full consistency with the geometric framework of Classical Continuum Mechanics~\cite{Nol1978,Rou1991a,Rou2006,KD2021}, we have to emphasize a slight difference with previous works in astrophysics concerning the choice of the volume form on $\body$.
In~\cite{Sou1958}, $V$ is equipped with the canonical 3-form $\vol_{\bq}=\dd X^{1} \wedge \dd X^{2} \wedge \dd X^{3}$ on $\RR^{3}$. In~\cite{CQ1972,BS2003,GHE2011}, the body $\body$ is equipped with a volume form $\Omega\in \Omega^{3}(\body)$ which represents the number density of conserved idealized particles (meant to be identified with the baryon number density in~\cite{CQ1972}). The three of them are, of course, proportional to each other on $\body$. As pointed out by Carter and Quintana, the Relativistic Hyperelasticity theory does not depend on the particular choice of a volume form on $\body$~\cite{CQ1972}. Our choice, here, of a volume form $\mu$, interpreted as a ``mass measure'' allows us to recover the mass densities encountered in Classical Continuum Mechanics, and to assimilate the integral
\begin{equation*}
  m=\int_{\body} \mu,
\end{equation*}
as the total mass of the continuous medium/mechanical structure under study.

\begin{rem}\label{rem:F-CCM}
  It is worth mentioning that one takes here a point of view reverse to the one of Classical Continuum Mechanics of solids~\cite{TN1965,LC1985,Rou1991a,MH1994,Rou2006}, in which a configuration is an embedding $p\colon\body \to \espace$ of the body $\body$ into the three-dimensional space $\espace$, endowed with the Euclidean metric $\bq$. In the present formalism, the main concept is a mapping $\Psi\colon \mM \to V$ from the Universe $\mM$ to the space of labels $V$. A key difference is that, in Classical Continuum Mechanics, $\pp$ and its tangent map, the so-called \emph{deformation gradient}
  \begin{equation*}
    \bF=T\pp \colon T\body \to T\espace
  \end{equation*}
  are invertible, whereas here, the matter field $\Psi$ and its tangent map $T\Psi$ are not.
\end{rem}

The pullback by $\Psi$ of the mass measure $\mu$ on the body $\body$
\begin{equation*}
  \omega:=\Psi^{*} \mu=(\mu\circ \Psi)(T\Psi \cdot, T\Psi \cdot, T\Psi \cdot)
\end{equation*}
is a 3-form defined on the four-dimensional World tube $\mW=\Psi^{-1}(\body)$. Since $T_{m}\Psi$ is assumed to be of rank $3$ at each point of $\mW$, there exists a \emph{nowhere vanishing} vector field $\bP$ on $\mW$, such that
\begin{equation}\label{eq:def-P}
  \omega=i_{\bP} \vol_{g},
\end{equation}
where $i_{\bP}$ is the interior product (or contraction) of $\omega$ by $\bP$. This vector field $\bP$ is the \emph{current of matter} (it was called \textit{vecteur courant de matière} in~\cite{Sou1958}).

\begin{rem}
  In 3D Classical Continuum Mechanics, the pushforward of the mass measure $\mu$ by the embedding $\pp\colon \body \to \espace$ \cite{KD2021}, when expressed using the 3D volume form $\vol_{q}$, is represented by a scalar density $\rho$ (indeed, $\pp_{*} \mu=\rho \, \vol_{q}$). In 4D, the pullback of the mass measure $\mu$ by the matter field $\Psi$, when expressed using the 4D volume form $\vol_{g}$, is represented by the quadrivector~$\bP$ (indeed, $\Psi^{*} \mu=i_{\bP} \vol_{g}$).
\end{rem}

At each point $m\in \mW$, the tangent vector $\bP(m)$ spans the one-dimensional subspace $\ker T_{m}\Psi$.
Indeed, the equality
\begin{equation*}
  \mu_{\Psi(m)}(T_{m}\Psi.\bP, T_{m}\Psi.\xi_{1}, T_{m}\Psi.\xi_{2}) = \vol_{g_{m}}(\bP(m),\bP(m),\xi_{1},\xi_{2})=0,
  \quad
  \forall  \xi_{1},\xi_{2} \in T_{m}\mW,
\end{equation*}
implies
\begin{equation}\label{eq:TPsi-P-null}
  T_{m}\Psi.\bP(m)=0,
\end{equation}
since $T_{m}\Psi$ is surjective. To describe perfect matter, Souriau assumes furthermore that $\bP$ is timelike, \emph{i.e.} that
\begin{equation*}
  \norm{\bP}_{g}^{2}=g(\bP, \bP)<0
\end{equation*}
on the World tube $\mW$ (we refer to~\cite{Sou1964} for the other cases, light for instance). Observe that $\bP$ defines a time orientation on $\mW$.

It will be proved as essential to define a unit timelike vector field $\bU$ collinear to $\bP$, and to write
\begin{equation}\label{eq:def-U}
  \bP=\rho_{r} \bU,
  \quad
  \text{with}
  \quad
  \norm{\bU}^{2}_{g}=-1.
\end{equation}
The function
\begin{equation}\label{eq:def-rho}
  \rho_{r} :=\sqrt{-\norm{\bP}_{g}^{2}} \, ,
\end{equation}
defined on the World tube $\mW$, is then interpreted as the \emph{rest mass density} \cite{Sou1958,Sou1964,KM1992}.

\begin{rem}\label{eq:cP}
  In Special Relativity, the vector field $c\bP$ is the \emph{four-momentum quadrivector}.
\end{rem}

We will finish this section by defining the \emph{conformation}, a fundamental concept introduced by Souriau in 1958. It is the cornerstone of the formulation of Relativistic Hyperelasticity at large scale, in particular, for the modeling of neutron stars with a solid crust. In recent works, it is sometimes referred to as \emph{strain}, but since \emph{this term} has a slightly different meaning in Classical Continuum Mechanics, we prefer to keep the initial name given by Souriau. The conformation is defined as the vector-valued function~\cite{Sou1958}
\begin{equation}\label{eq:conformation}
  \bH: \mM \to \Sym^{2}V, \qquad m \mapsto \bH(m):= (T_{m}\Psi)\,g_{m}^{-1} (T_{m}\Psi)^{\star},
\end{equation}
where $\Sym^{2}V$ is the six-dimensional vector space of symmetric contravariant second-order tensors on $V$. In simpler words, $\bH$ is a function from the four-dimensional manifold $\mM$ to the vector space of $3 \times 3$ matrices. The hypothesis we made that $\Psi$ is a submersion on $\mW$, together with the hypothesis that $\ker T\Psi$ is generated by the timelike vector field $\bU$ implies that $\bH(m)$ is positive definite for all $m \in \mW$.

\begin{rem}\label{rem:Psim1}
  Since the mapping $\Psi: \mW \to \body$ is not invertible, the conformation cannot be considered, \textit{stricto sensu}, as the pushforward of $g^{-1}$, which is not defined. It is thus not, strictly speaking, a co-metric on $\body$, but a vector-valued function of $m \in \mW$ with value a symmetric second-order contravariant tensor in $V$. Note that, if we forget that $\body$ is a domain in the vector space $V$ but consider that it is a manifold, then $\bH$ is interpreted as a tensor field along $\Psi$ with values in $\Sym^{2} T\body$, in other words, it is a section of the pullback bundle $\Psi^{*}(\Sym^{2} T\body)$.
\end{rem}

% ----------------------------------------------------------------
\section{Conservation of matter}
\label{sec:mass-conservation}
% ----------------------------------------------------------------

Since the exterior derivative of the mass measure $\mu$ on the 3-dimensional manifold $\body$ vanishes, $\dd \mu=0$, we get the following conservation law.

\begin{lem}[Souriau, 1958] \label{lem:divP}
  We have the following conservation law on the World tube $\mW$
  \begin{equation*}
    \dive^{g} \bP=0.
  \end{equation*}
\end{lem}

\begin{proof}
  Let $\Lie_{\bP}$ be the Lie derivative with respect to $\bP$. Then,
  \begin{equation*}
    \Lie_{\bP} \vol_{g} = (\dive^{g} \bP)\, \vol_{g},
  \end{equation*}
  but, using Cartan magic formula,
  \begin{equation*}
    \Lie_{\bP} \vol_{g} = \dd \, i_{\bP} \vol_{g}=\dd\Psi^{*} \mu=\Psi^{*} \dd\mu=0.
  \end{equation*}
\end{proof}

\begin{rem}
  In Special Relativity, the equation
  \begin{equation*}
    \dive^{g} \bP = \dive^{g} (\rho_{r} \bU)=0
  \end{equation*}
  recasts as the usual continuity equation of Classical Fluid Dynamics \cite{Eck1940,Sou1958,Sou1964}, and is interpreted as the Relativistic mass conservation. It will be shown, furthermore, in \autoref{sec:Gallilean}, that $ \dive^{g} c \bP=0$ converges towards the classical continuity equation under subsequent hypothesis.
\end{rem}

If the body $\body$ is endowed with a Riemannian metric $\bgamma_{0}$, the rest mass density $\rho_{r}$ can be related to the conformation $\bH$, as demonstrated by Souriau in~\cite{Sou1958}, where he chose $\bgamma_{0}:=\bq$, the canonical Euclidean metric on $V=\RR^{3}$ (see also~\cite{KM1992}). The notation $\bgamma_{0}$ is chosen for consistency with the intrinsic geometric framework of three-dimensional Hyperelasticity~\cite{Nol1978,Rou1991a,Rou2006,KD2021}, and $\bgamma_{0}$ is not necessarily equal to $\bq$, as discussed by several authors~\cite{BS2003,GHE2011} (see \autoref{sec:reference-metric} for a discussion about different choices for $\bgamma_{0}$).

\begin{lem}[Souriau, 1958] \label{lem:PCM_GR}
  Let $\bgamma_{0}$ be a fixed Riemannian metric on the body $\body$. Then, the rest mass density $\rho_{r}$ can be expressed as
  \begin{equation}\label{eq:PCM_GR}
    \rho_{r} = (\rho_{\bgamma_{0}} \circ \Psi) \sqrt{\det\left[ \bH(\bgamma_{0} \circ \Psi)\right]},
  \end{equation}
  where $\Psi$ is the matter field, $\bH$ is the conformation, and
  \begin{equation*}
    \rho_{\bgamma_{0}} = \frac{\mu}{\vol_{\bgamma_{0}}}.
  \end{equation*}
\end{lem}

\begin{rem}\label{rem:PCM_GR}
  The function $\rho_{\bgamma_{0}}$ is defined on the body $\body$, and interpreted as the mass density with respect to the Riemannian volume form $\vol_{\bgamma_{0}}$. It is very important to note, for subsequent applications, that $\rho_{\bgamma_{0}}$ is independent of the metric $g$ on the Universe $\mM$. Moreover, one can check that the right hand-side of~\eqref{eq:PCM_GR} does not depend on $\bgamma_{0}$, as expected. Indeed if we substitute $\bgamma_{1}$ to $\bgamma_{0}$ in~\eqref{eq:PCM_GR}, one has $\rho_{\bgamma_{1}} = \sqrt{\det(\bgamma_{1}^{-1}\bgamma_{0})}\,\rho_{\bgamma_{0}}$, whereas $\det(\bH\bgamma_{1}) = \det(\bgamma_{0}^{-1}\bgamma_{1}) \det( \bH\bgamma_{0})$, and thus
  \begin{equation*}
    (\rho_{\bgamma_{0}} \circ \Psi) \sqrt{\det\left[ \bH(\bgamma_{0} \circ \Psi)\right]} = (\rho_{\bgamma_{1}} \circ \Psi) \sqrt{\det\left[ \bH(\bgamma_{1} \circ \Psi)\right]}.
  \end{equation*}
\end{rem}

\begin{proof}
  Note first that there exists a function $\rho_{\bgamma_{0}}$ (a mass density) defined on $\body$ such that
  \begin{equation*}
    \mu = \rho_{\bgamma_{0}} \vol_{\bgamma_{0}}.
  \end{equation*}
  Thus, by~\eqref{eq:def-P}--\eqref{eq:def-U}, we get
  \begin{equation*}
    \omega = \rho_{r} \, i_{\bU}\vol_{g} = \Psi^{*}\mu = (\rho_{\bgamma_{0}}\circ\Psi) \Psi^{*}\vol_{\bgamma_{0}}.
  \end{equation*}
  Let $m \in \mW$ and $(\ee_{0} = \bU(m),\ee_{i})$ be a direct orthonormal basis of $T_{m}\mM$. Then, we have
  \begin{equation*}
    \vol_{g_{m}}(\bU(m),\ee_{1},\ee_{2},\ee_{3}) = 1
  \end{equation*}
  and
  \begin{equation*}
    \rho_{r}(m) = \rho_{r}(m)\vol_{g_{m}}(\bU(m),\ee_{1},\ee_{2},\ee_{3}) = \rho_{\bgamma_{0}}(\Psi(m)) \, \vol_{\bgamma_{0}(\Psi(m))}(T_{m}\Psi \ee_{1}, T_{m}\Psi \ee_{2},T_{m}\Psi \ee_{3}).
  \end{equation*}
  Observe now that the vector space $T_{\Psi(m)}\body$ is endowed with two Euclidean structures; the first one, defined by $\bgamma_{0}(\Psi(m))$ and the second one, defined by $\bH(m)^{-1}$. Besides, the restriction of $T_{m}\Psi$ to the three-dimensional subspace $\bU(m)^{\bot}$ (the orthogonal complement of $\bU(m)$ in $T_{m}\mM$) is a linear isomorphism and
  \begin{equation*}
    T_{m}\Psi\colon (\bU(m)^{\bot},g_{m}) \to (T_{\Psi(m)}\body,\bH(m)^{-1})
  \end{equation*}
  is an isometry, by the very definition of the conformation $\bH$. Hence, $(T_{m}\Psi \ee_{1}, T_{m}\Psi \ee_{2},T_{m}\Psi \ee_{3})$ is a direct orthonormal basis of the Euclidean space $(T_{\Psi(m)}\body,\bH(m)^{-1})$ and thus
  \begin{equation*}
    \vol_{\bH(m)^{-1}}(T_{m}\Psi \ee_{1}, T_{m}\Psi \ee_{2},T_{m}\Psi \ee_{3}) = 1.
  \end{equation*}
  We have therefore
  \begin{align*}
    \vol_{\bgamma_{0}(\Psi(m))}(T_{m}\Psi \ee_{1}, T_{m}\Psi \ee_{2},T_{m}\Psi \ee_{3})
     & = \sqrt{\det\left[\bH(m)\bgamma_{0}(\Psi(m))\right]}\, \vol_{\bH(m)^{-1}}(T_{m}\Psi \ee_{1}, T_{m}\Psi \ee_{2},T_{m}\Psi \ee_{3}) \\
     & =  \sqrt{\det\left[\bH(m)\bgamma_{0}(\Psi(m))\right]},
  \end{align*}
  and thus
  \begin{equation*}
    \rho_{r}(m) = \rho_{\bgamma_{0}}(\Psi(m)) \, \sqrt{\det\left[\bH(m)\bgamma_{0}(\Psi(m))\right]}.
  \end{equation*}
\end{proof}

\begin{rem}\label{rem:PCM_Omega0}
  In Classical Continuum Mechanics, the body is often identified with a reference configuration $\Omega_{0}$ embedded in $\RR^{3}$ and endowed with the euclidean metric $\bgamma_{0}=q$. Two mass densities, $\rho_{0}$ on $\Omega_{0}$ and $\rho$ on the deformed configuration $\Omega$ (also embedded in $\RR^{3}$), are usually defined. A classical expression of mass balance is formulated on $\Omega_{0}$ as
  \begin{equation}\label{eq:PCM_Omega0}
    \rho_{0} = (\rho \circ \phi) \sqrt{\det(\bq^{-1}\bC)},
  \end{equation}
  where $\phi \colon \Omega_{0}\to \Omega$ is the \emph{deformation}, $\bC:=\phi^{*}\bq$ is the \emph{right Cauchy--Green tensor} (defined on $\Omega_{0}$ as the pullback by the deformation $\phi$ of the Euclidean metric $\bq$). The formal comparison of~\eqref{eq:PCM_GR}, recast as
  \begin{equation*}
    \rho_{\bgamma_{0}}\circ \Psi=\rho_{r}  \sqrt{\det\left[ (\bgamma_{0}^{-1} \circ \Psi)\bH^{-1}\right]},
  \end{equation*}
  with ~\eqref{eq:PCM_Omega0}, shows that~\eqref{eq:PCM_GR} can be interpreted as a Relativistic generalization of the mass conservation law for Galilean deformable solids. It also shows that the inverse of the (contravariant) conformation $\bH$ plays the role of the (covariant) right Cauchy--Green tensor $\bC$.
\end{rem}

% ----------------------------------------------------------------
\section{Conformation and strains}
\label{sec:Conformation-strains}
% ----------------------------------------------------------------

The existence of the unit timelike vector field $\bU$ on the World tube $\mW$ allows to perform the related orthogonal decompositions of the metric $g$ and co-metric $g^{-1}$ (see~\autoref{sec:Orth-Decomp-2nd}),
\begin{equation}\label{eq:def-h}
  g = h-\bU^\flat \otimes \bU^{\flat},
  \qquad
  g^{-1} = h^{\sharp}-\bU \otimes \bU,
  \qquad \text{on $\mW$},
\end{equation}
where the tensor fields $h$ (noted $E$ in \cite{KM1997}) and $h^{\sharp}=g^{-1} h g^{-1}$, the spatial part of $g$ and $g^{-1}$ respectively, are uniquely defined by the conditions
\begin{equation}\label{eq:hU}
  h\bU = 0, \quad \text{and} \quad h = g \quad \text{on} \quad  \bU^{\perp},
\end{equation}
where $\bU^{\perp}$ is the three-dimensional (necessarily spacelike) orthogonal subbundle to $\bU$. Both $h$ and $h^{\sharp}$ have signature $(0,+,+,+)$. These orthogonal decompositions are highlighted at the beginning of most works on Relativistic Fluids or Solids \cite{Eck1940,Lic1955,CQ1972,KM1997}. We point out, however, that Souriau did not need to perform them to derive the general covariant formulation of Relativistic Hyperelasticity \cite{Sou1958,Sou1964}. There are two reasons for it. First, the four-dimensional symmetric second-order tensors $h$ and $h^{\sharp}$ are strongly related to the conformation
\begin{equation*}
  \bH = (T\Psi)\,g^{-1} (T\Psi)^{\star},
\end{equation*}
by lemma~\ref{lem:h-H}. Secondly, $h$ and $h^{\sharp}$ do not appear naturally in the derivation of a general covariant formulation of Relativistic Hyperelasticity, contrary to the conformation $\bH$ (see theorem~\ref{thm:LPsiH}).

\begin{lem}\label{lem:h-H}
  On the World tube $\mW$, we have
  \begin{equation}\label{eq:Hhsharp}
    \bH = (T\Psi)\,h^{\sharp} (T\Psi)^{\star},
    \quad \text{and}\quad
    h=(T\Psi)^{\star}\, \bH^{-1} T\Psi.
  \end{equation}
  where $h=g+\bU^\flat \otimes \bU^{\flat}$ and $h^{\sharp}=g^{-1} h g^{-1}$.
\end{lem}

\begin{proof}
  First, since $T\Psi \bU=0$, the conformation, when restricted to the World tube $\mW$, recasts as
  \begin{equation*}
    \bH = (T\Psi)\left(h^{\sharp} - \bU \otimes \bU \right)(T\Psi)^{\star}=(T\Psi)\,h^{\sharp} (T\Psi)^{\star}.
  \end{equation*}
  Then, to prove the second equality, remark that the statement is pointwise. Therefore, we can use an orthonormal basis $(\ee_{\mu})$ of $T_{m}\mM$ with $\ee_{0} = \bU(m)$. In this basis, $h_{m}$ is represented by the $4 \times 4$ matrix
  \begin{equation*}
    \begin{pmatrix}
      0 & 0     \\
      0 & I_{3}
    \end{pmatrix},
  \end{equation*}
  where $I_{3}$ is the $3 \times 3$ identity matrix. Now, respectively to this basis and the canonical basis of $\RR^{3}$, the linear map $T_{m}\Psi\colon T_{m}\mM \to \RR^{3}$ is represented by the matrix
  \begin{equation*}
    \begin{pmatrix}
      0 & M \\
    \end{pmatrix},
  \end{equation*}
  where $M$ is a $3 \times 3$ invertible matrix, and its transpose $(T_{m}\Psi)^{\star}$ by the matrix
  \begin{equation*}
    \begin{pmatrix}
      0         \\
      M^{\star} \\
    \end{pmatrix}.
  \end{equation*}
  Thus, we have
  \begin{equation*}
    \bH(m) = (T_{m}\Psi)\,g_{m}^{-1} (T_{m}\Psi)^{\star} =
    \begin{pmatrix}
      0 & M \\
    \end{pmatrix}
    \begin{pmatrix}
      -1 & 0     \\
      0  & I_{3}
    \end{pmatrix}
    \begin{pmatrix}
      0         \\
      M^{\star} \\
    \end{pmatrix}
    = MM^{\star},
  \end{equation*}
  and
  \begin{equation*}
    (T_{m}\Psi)^{\star}\, \bH(m)^{-1} T_{m}\Psi =
    \begin{pmatrix}
      0         \\
      M^{\star} \\
    \end{pmatrix}
    M^{-\star}M^{-1}
    \begin{pmatrix}
      0 & M \\
    \end{pmatrix}
    =
    \begin{pmatrix}
      0     \\
      I_{3} \\
    \end{pmatrix}
    \begin{pmatrix}
      0 & I_{3} \\
    \end{pmatrix}
    = h_{m}.
  \end{equation*}
\end{proof}

The definition of a strain in (hyper)elasticity is usually obtained by comparing two metrics. If the body $\body$ is endowed with a fixed Riemannian metric $\bgamma_{0}$, it can be used to define a strain tensor in Relativistic Hyperelasticity. A first possibility~\cite{CQ1972,Mau1978b} is to introduce the pullback by $\Psi$ of $\bgamma_{0}$,  given by
\begin{equation}\label{eq:h0}
  h_{0} := \Psi^{*}\bgamma_{0}=(T\Psi)^{\star} (\bgamma_{0}\circ \Psi) T\Psi,
\end{equation}
and called a \emph{frozen metric} in~\cite{KM1992,KM1997} (these authors note it $h$ rather than $h_{0}$).
It is defined on the World tube $\mW$ and is of signature $(0, +, +, +)$. Conversely, given a quadratic form $h_{0}$ on $\mW$ with signature $(0, +,+,+)$, the question of when it can be realized as the pullback by $\Psi$ of a fixed Riemannian metric $\bgamma_{0}$ on the body, has been investigated by Kijowski and Magli~(see \autoref{sec:reference-metric}).

A possible generalization of the \emph{Euler-Almansi strain tensor}~\cite{CQ1972,Mau1978b} is then obtained as the four-dimensional symmetric covariant tensor field,
\begin{equation}\label{eq:strains-h}
  \be := \frac{1}{2} (h-h_{0}).
\end{equation}
Note that $\be=0$ for $h=h_{0}$ and that $\be$ is degenerate since $\be \bU=0$.

\begin{rem}
  As observed by Carter and Quintana \cite{CQ1972}, since the linear tangent map $T\Psi$ plays a role similar to that of the inverse of the tangent map $\bF=T\pp$ in Classical Continuum Mechanics (see remark~\ref{rem:F-CCM}), the frozen metric $h_{0}$ plays a role similar to that of the inverse, sometimes called the \emph{finger deformation tensor}, of the \emph{left Cauchy--Green tensor} $\bb:=\bF\bgamma_{0}^{-1} \bF^{\star}$.
\end{rem}

Other choices for strain tensors similar to the ones of Classical Continuum Mechanics can be made, for instance the following ones which are simpler and probably more relevant,
\begin{equation}\label{eq:strains-H}
  \mathfrak E:= \frac{1}{2} \left(\bH^{-1}-\bH_{0}^{-1}\right)
  \quad \text{or} \quad
  \widehat{\mathfrak E} := -\frac{1}{2}\log \big(\bH\,\bH_{0}^{-1}\big),
\end{equation}
where
\begin{equation} \label{eq:H0}
  \bH_{0}:=\bgamma_{0}^{-1} \circ \Psi.
\end{equation}
The first one generalizes the \emph{Green--Lagrange strain}, whereas the second one generalizes the \emph{logarithmic strain} introduced by Becker~\cite{Bec1893} and Hencky~\cite{Hen1928} (see~\cite{MMEN2018}). They both vanish when $\bH^{-1}=\bH_{0}^{-1}=\bgamma_{0}\circ \Psi$. These strain tensors are three-dimensional second-order tensors. Like the conformation, they are not tensor fields on $\body$ but vector valued functions defined on the World tube $\mW$ with values in $\Sym^{2}V$.

Note that $\bH_{0}$ is related to $h_{0}$ by
\begin{equation*}
  h_{0}=(T\Psi)^{\star}\,\bH_{0} \, (T\Psi)
\end{equation*}
and that, by lemma~\ref{lem:h-H} and definitions~\eqref{eq:h0}--\eqref{eq:strains-h}, $\mathfrak E$ is connected to $\be$, defined by~\eqref{eq:strains-h}, by
\begin{equation*}
  \be=(T\Psi)^{\star}\, \mathfrak E\, (T\Psi)
  \quad
  \text{on $\mW$}.
\end{equation*}

% ----------------------------------------------------------------
\section{Lagrangian formulation}
\label{sec:Lagrangian-formulation}
% ----------------------------------------------------------------

In~\cite{Sou1958,Sou1964}, Souriau has proposed a clear and detailed formulation of Hyperelasticity in the framework of General Relativity. He called this formulation \emph{Variational Relativity} (which is the title of~\cite{Sou1958}). His approach consists in writing Lagrangians (\textit{i.e.} functionals depending on tensorial fields) and looking for critical points of them (Principle of Least, or Stationary, Action). This formulation is inspired by Gauge Theory~\cite{Ble1981}, which is the main framework of Fields Theory and Quantum Mechanics and can also be used to formulate General Relativity using variational principles (see Palatini's Method~\cite{Pal1919,FFR1982}).

The starting point is the \emph{Hilbert-Einstein functional}
\begin{equation}\label{eq:Hilbert-Einstein}
  \mH(g) = \int (a R_{g} + b)\, \vol_{g},
\end{equation}
defined formally on the set of all Lorentzian metrics on the Universe $\mM$. Here, the two constants $a$ and $b$ are related to the Einstein constant $\kappa$ (depending on the Newton constant $G$) and the cosmological constant $\Lambda$ by
\begin{equation*}
  \kappa = \frac{8\pi G}{c^{4}} = \frac{1}{2a}, \qquad \Lambda = -\frac{b}{2a}.
\end{equation*}
As derived first by Hilbert~\cite{Hil1924}, the $L^{2}$-gradient of $\mH$ (for Ebin's metric~\cite{Ebi1968}) is the symmetric second order covariant tensor field
\begin{equation}\label{eq:H-gradient}
  \grad \mH = a\, \Ric_{g}-\frac{1}{2}(a\, R_{g}+b)g = \frac{1}{2\kappa}\left( \Ein_{g} + \Lambda g\right),
\end{equation}
where $\Ric_{g}$ is the Ricci tensor of the metric $g$, $R_{g} = \tr(g^{-1}\Ric_{g})$ is the scalar curvature, and $\Ein_{g}$ is the \emph{Einstein tensor}, defined by
\begin{equation}\label{eq:Einstein-tensor}
  \Ein_{g} := \Ric_{g}-\frac{1}{2}R_{g}\, g.
\end{equation}

The critical points of $\mH$ are the solutions of Einstein's equation in the vacuum (with cosmological constant)
\begin{equation*}
  \Ein_{g} + \Lambda g = 0.
\end{equation*}

To introduce the effects of matter in this framework, a second functional $\mL^{\text{matter}}(g, \Psi)$, depending on the metric $g$ and the matter field $\Psi$, is added to $\mH$ to build a new Lagrangian
\begin{equation*}
  \mL(g,\Psi) = \mH(g) + \mL^{\text{matter}}(g, \Psi).
\end{equation*}
Following Souriau~\cite{Sou1958,Sou1964}, for Relativistic continua, one assumes that the Lagrangian for perfect matter $\mL^{\text{matter}}(g, \Psi)$ depends only on the $0$-jet of the metric $g$ and of the $1$-jet of the matter field $\Psi$. In other words, it takes the form
\begin{equation}\label{eq:LagrangianFirstJets}
  \mL^{\text{matter}}(g, \Psi) = \int L_{0}\left(g_{\mu\nu}, \Psi^{I}, \frac{\partial \Psi^{I}}{\partial x^{\mu}} \right) \vol_{g},
\end{equation}
where
\begin{equation*}
  L_{0} \colon (\bgamma,\vv,\btau) \to L_{0}(\bgamma,\vv,\btau)
\end{equation*}
is a smooth scalar function which has for arguments a quadratic form $\bgamma$ (of signature $(-,+,+,+)$) on $\RR^{4}$, a vector $\vv \in \RR^{3}$ and matrix $\btau$ with 3 raws and 4 columns. The function $L_{0}$ is called the \emph{Lagrangian density} of the functional $\mL^{\text{matter}}$, and its evaluation on the fields $(g,\Psi)$, that is $L_{0}\left(g_{\mu\nu}, \Psi^{I}, \frac{\partial \Psi^{I}}{\partial x^{\mu}} \right)$, will be denoted as $L_{0}(g, \Psi, T\Psi)$. Its evaluation at a point $m\in \mM$ is then noted $L_{0}(g_{m}, \Psi(m), T_{m}\Psi)$.

\begin{rem}
  The Lagrangian density $L_{0}$ is noted $p$ and called the \emph{presence} in~\cite{Sou1958}. It is noted $\epsilon$ and called the \emph{rest frame energy density} in~\cite{KM1992}. It is noted $\rho$ or $\sigma$ in \cite{BS2003}.
\end{rem}

In order to avoid unnecessary analytical difficulties and since, in practice, we do not require that Lagrangian densities are integrable over the whole manifold $\mM$, usually not compact, Lagrangian densities are integrated only over relatively
compact domains $U$ (and furthermore contained in a local chart). Therefore, we shall write
\begin{equation*}
  \mL_{U}(g, \Psi) = \int_{U} L(g_{m}, \Psi(m), T_{m}\Psi)\, \vol_{g},
\end{equation*}
to emphasize the dependence on $U$. When we just want to express that a Lagrangian $\mL$ is defined by the Lagrangian density
$L$, we simply write
\begin{equation*}
  \mL(g, \Psi) = \int L(g_{m}, \Psi(m), T_{m}\Psi) \, \vol_{g},
\end{equation*}
omitting the domain of integration.

The main postulate of General Relativity is precisely that \emph{Physical laws must be independent of the choice of coordinates}. This principle is known as \emph{General Covariance}, or invariance by coordinates change, or invariance by (local) diffeomorphisms. Let us describe this principle in more precise terms and formulate its consequences. Let $\varphi \colon U \to \widetilde{U}$ be a diffeomorphism between two open sets $U$ and $\widetilde{U}$. Then, the Lagrangian $\mL$ is invariant by $\varphi$ if
\begin{equation}\label{eq:invmL}
  \mL_{U}(\varphi^{*} g, \varphi^{*}\Psi) = \mL_{\widetilde{U}}(g, \Psi),
\end{equation}
for every Lorentzian metrics $g$ on $\mM$, and vector valued functions $\Psi:\mM\to V$. Here, the action of a (local) diffeomorphism $\varphi$ on these field variables is defined by
\begin{equation*}
  \varphi^{*} g = (T\varphi)^{\star}\,(g \circ \varphi) (T\varphi), \quad \text{and} \quad \varphi^{*} \Psi=\Psi \circ \varphi .
\end{equation*}
If the invariance~\eqref{eq:invmL} holds for every local diffeomorphism $\varphi \colon U \to \widetilde{U}$, then, $\mL$ is said to be \emph{general covariant}.

\begin{rem}\label{rem:H-invariance}
  It is well-known that the Hilbert-Einstein functional $\mH$ is general covariant. Indeed,
  \begin{equation*}
    \mH_{U}(\varphi^{*} g) = \mH_{\widetilde{U}}(g),
  \end{equation*}
  for every diffeomorphism $\varphi\colon U \to \widetilde{U}$, by virtue of the \emph{change of variables formula}
  \begin{equation*}
    \int_{U} \varphi^{*} \omega = \int_{\widetilde{U}} \omega,
  \end{equation*}
  and because
  \begin{equation*}
    (a R_{\varphi^{*} g} + b)\, \vol_{\varphi^{*} g} = \varphi^{*}[(a R_{g} + b)\, \vol_{g}].
  \end{equation*}
  As shown by Noether \cite{Noe1918,Kos2011}, a direct consequence of this invariance is the fundamental property~\cite{Ein1915,Wey1917}
  \begin{equation*}
    \dive^{g}(\Ein_{g} + \Lambda g) = \dive^{g} \Ein_{g} = 0.
  \end{equation*}
\end{rem}

\begin{lem}
  If the Lagrangian
  \begin{equation*}
    \mL^{\text{matter}}(g, \Psi) = \int L_{0}(g_{m}, \Psi(m), T_{m}\Psi) \, \vol_{g}
  \end{equation*}
  is general covariant, then, its Lagrangian density satisfies
  \begin{equation}\label{eq:Lagrangian-density-covariance}
    L_{0}(\bA^{\star} \bgamma \,\bA, \vv, \btau\bA) = L_{0}(\bgamma, \vv, \btau),
    \quad
    \forall \bA \in \GL(4).
  \end{equation}
\end{lem}

\begin{proof}
  Let $\varphi \colon U \to \widetilde{U}$ be a diffeomorphism between two open sets $U$ and $\widetilde{U}$ and set
  \begin{equation*}
    f(\yy) := L_{0}(A(\yy)^{\star}g_{\yy} A(\yy),\Psi(\yy), T_{\yy}\Psi A(\yy))
  \end{equation*}
  for $\yy \in \widetilde{U}$, where $A(\yy) = T_{\varphi^{-1}(\yy)}\varphi$. Then, $T_{m}\varphi = A(\varphi(m))$, for $m \in U$ and
  \begin{equation*}
    L_{0}((\varphi^{*}g)_{m}, (\varphi^{*}\Psi)(m), T_{m}(\varphi^{*}\Psi)) = f(\varphi(m)).
  \end{equation*}
  Therefore
  \begin{equation*}
    \mL_{U}(\varphi^{*} g, \varphi^{*}\Psi) = \int_{U} f(\varphi(m)) \varphi^{*} \vol_{g} = \int_{\widetilde{U}} f(\yy) \vol_{g},
  \end{equation*}
  by the change of variable formula, and the general covariance property leads to
  \begin{equation*}
    f(\yy) = L_{0}(A(\yy)^{\star}g_{\yy} A(\yy),\Psi(\yy), T_{\yy}\Psi A(\yy)) = L_{0}(g_{\yy},\Psi(\yy), T_{\yy}\Psi), \quad \forall \yy.
  \end{equation*}
  Hence, the Lagrangian density is subject to the following invariance
  \begin{equation*}
    L_{0}(\bA^{\star} \bgamma \bA, \vv, \btau \bA) = L_{0}(\bgamma, \vv, \btau),
    \quad
    \forall \bA \in \GL(4).
  \end{equation*}
\end{proof}

\begin{rem}
  Since the Lie derivative is the infinitesimal version of the pullback, meaning that
  \begin{equation*}
    \Lie_{X}\bT := [\partial_{s}\varphi(s)^{*}\bT]_{s=0}
  \end{equation*}
  for every tensor field $\bT$ and every path of (local) diffeomorphisms $\varphi(s)$ with
  \begin{equation*}
    \varphi(0)=\id, \quad \text{and} \quad [\partial_{s}\varphi(s)]_{s=0}=X,
  \end{equation*}
  there is also an almost\footnote{indeed equivalent to covariance by diffeomorphisms isotopic to the identity.} equivalent \emph{infinitesimal formulation of general covariance}~\cite{Noe1918}, which is used by several authors (such as in \cite{Wer2006}). For instance, in the present case, the general covariance of the matter Lagrangian $\mL^{\text{matter}}$
  \begin{equation*}
    \mL^{\text{matter}}_{U}(\varphi^{*} g, \varphi^{*}\Psi) = \mL^{\text{matter}}_{\widetilde{U}}(g, \Psi),
  \end{equation*}
  for every local diffeomorphism $\varphi \colon U \to \widetilde{U}$ leads to
  \begin{equation*}
    \frac{\delta\mL^{\text{matter}}_{U}}{\delta g} . \Lie_{X}g + \frac{\delta\mL^{\text{matter}}_{U}}{\delta \Psi}.\Lie_{X}\Psi = 0.
  \end{equation*}
  Therefore, its Lagrangian density must satisfy (see~\cite{Sop2008})
  \begin{equation*}
    \frac{\partial L_{0}}{\partial \bgamma}:(\Lie_{X}g)_{m} + \frac{\partial L_{0}}{\partial \vv}\cdot (\Lie_{X}\Psi)(m) + \frac{\partial L_{0}}{\partial \btau}:(T_{m}\Lie_{X}\Psi) = 0, \qquad \forall m.
  \end{equation*}
\end{rem}

The following result is essential for the formulation of Relativistic Hyperelasticity and exhibits the fundamental role played by the conformation. It must be compared to the fact that an elastic energy in Classical Continuum Mechanics, which is \emph{objective} (\textit{i.e.} satisfies the \emph{material frame indifference} principle~\cite{TN1965}) depends on the deformation $\varphi$ only through the right Cauchy--Green tensor $\bC =\varphi^{*}q$.

\begin{thm}[Souriau (1958)]\label{thm:LPsiH}
  Suppose that the Lagrangian
  \begin{equation*}
    \mL^{\text{matter}}(g, \Psi) = \int L_{0}(g_{m}, \Psi(m), T_{m}\Psi) \, \vol_{g}
  \end{equation*}
  is general covariant. Then, its Lagrangian density can be written as
  \begin{equation*}
    L_{0}(g, \Psi, T\Psi) = L(\Psi, \bH),
  \end{equation*}
  for some function $L$, where $\bH= (T\Psi)\,g^{-1} (T\Psi)^{\star}$ is the conformation.
\end{thm}

The proof provided below is simpler and shorter that the original one given by Souriau in~\cite{Sou1958}. The reason for it is that, in this paper, we consider only perfect matter, in which case the conformation $\bH$ is positive definite at each point $m$ of the World tube.
This is not an hypothesis which is made in~\cite{Sou1958}.

\begin{proof}
  Consider a smooth Lagrangian density $L_{0}(\bgamma, \vv, \btau)$, where $\bgamma$ is a quadratic form of signature $(-,+,+,+)$ on $\RR^{4}$, $\vv \in \RR^{3}$ and $\btau\in \mathcal{L}(\RR^{4},\RR^{3})$ satisfies $\btau\overline{\btau} >0$. Suppose moreover that this Lagrangian density satisfies the following covariance property
  \begin{equation*}
    L_{0}(\bA^{\star} \bgamma \bA, \vv, \btau \bA) = L_{0}(\bgamma, \vv, \btau),  \qquad \forall \bA\in \GL(4).
  \end{equation*}

  First, we can find $\bA\in \GL(4)$ such that $\bA^{\star} \bgamma \bA=\eta$, where
  \begin{equation*}
    \eta =
    \begin{pmatrix}
      -1 & 0 \\
      0  & q
    \end{pmatrix}
  \end{equation*}
  is the canonical Lorentz inner product. Hence we get
  \begin{equation*}
    L_{0}(\bgamma, \vv, \btau) = L_{0}(\eta, \vv, \btau_{1}), \quad \text{with} \quad \btau_{1} = \btau \bA
    \quad \text{and} \quad \bA^{\star} \bgamma \bA = \eta.
  \end{equation*}

  Now, we introduce the following change of variables $\btau_{1} \mapsto (R, \bH)$, where
  \begin{equation*}
    \bH = \btau_{1}\eta^{-1} \btau_{1}^{\star} = \btau\bgamma^{-1}\btau^{\star}, \qquad R = V^{-1}\btau_{1},
  \end{equation*}
  and $V$ is the positive square root of the positive definite symmetric operator on $(\RR^{3},q)$
  \begin{equation*}
    \btau_{1}\overline{\btau_{1}} = \bH q =  \btau_{1}\eta^{-1}\btau_{1}^{\star}q
  \end{equation*}
  with $q$, the canonical Euclidean metric on $\RR^{3}$.

  We can check that $R\overline{R} = R \eta^{-1}R^{\star}q = I_{3}$ is a condition which defines a submanifold of the vector space of linear mappings $\mathcal{L}(\RR^{4},\RR^{3})$, and that $\btau_{1} \mapsto (R, \bH)$ is a diffeomorphism from the open set
  \begin{equation*}
    \set{\btau_{1} \in \mathcal{L}(\RR^{4},\RR^{3});\; \btau_{1}\overline{\btau_{1}} > 0}
  \end{equation*}
  onto the manifold
  \begin{equation*}
    \set{R \in \mathcal{L}(\RR^{4},\RR^{3});\; R\overline{R} = I_{3}} \times \set{\bH \in \Sym^{2}(\RR^{3});\; \bH >0}.
  \end{equation*}
  Hence, we can find a smooth function $L_{1}(\vv, R, \bH)$, such that
  \begin{equation*}
    L_{0}(\eta, \vv, \btau_{1}) = L_{1}(\vv, R, \bH),
  \end{equation*}
  with the property that
  \begin{equation*}
    L_{1}(\vv, R, \bH) = L_{0}(\eta, \vv, \btau_{1}) = L_{0}(\eta, \vv, \btau_{1}Q) = L_{1}(\vv, RQ, \bH),
  \end{equation*}
  for every Lorentz transformation $Q$. Next, we can find a Lorentz transformation $Q$ such that $RQ = R_{0}$ with
  \begin{equation*}
    R_{0} =
    \begin{pmatrix}
      0 & I_{3}
    \end{pmatrix},
  \end{equation*}
  because $R\overline{R} = R_{0}\overline{R_{0}} = I_{3}$. Therefore, we get finally
  \begin{equation*}
    L_{0}(\bgamma, \vv, \btau) = L_{0}(\eta, \vv, \btau_{1}) =  L_{1}(\vv, R, \bH) = L_{1}(\vv, R_{0}, \bH),
  \end{equation*}
  and $L_{1}(\vv, R_{0}, \bH)$ is a function $L(\vv, \bH)$, which depends only on $\vv$ and $\bH = \btau_{1}\eta^{-1} \btau_{1}^{\star} = \btau\bgamma^{-1}\btau^{\star}$.
\end{proof}

The following splitting of the Lagrangian density has been introduced by Souriau~\cite{Sou1958,Sou1964} and DeWitt~\cite{deW1962}:
\begin{equation}\label{eq:LSouriau}
  L(\Psi, \bH) = \rho_{r} c^{2}+E(\Psi, \bH)=\rho_{r} c^{2}+\rho_{r} e(\Psi, \bH),
\end{equation}
where $\rho_{r}$ is the rest mass density, expressed as
\begin{equation*}
  \rho_{r}=\rho_{\bgamma_{0}}(\Psi) \sqrt{\det\left[ \bH(\bgamma_{0} \circ \Psi)\right]},
\end{equation*}
by lemma~\ref{lem:PCM_GR}, provided a fixed metric $\bgamma_{0}$ has been given on the body $\body$ and $\rho_{\bgamma_{0}}=\mu/\vol_{\bgamma_{0}}$. The contribution $\rho_{r} c^{2}$ alone ($E=0$) allows for the modeling of perfect (non electromagnetic) dust. The function $E$ (resp. $e$) is the \emph{internal energy density} (resp. the \emph{specific  internal energy}). It is representative of perfect fluids when its dependency on $\bH$ is introduced only through the determinant $\det\left[ \bH(\bgamma_{0} \circ \Psi)\right]$. The additional dependency on $\Psi$ and $\bH$ through the energy density $E$ is more generally representative of Relativistic hyperelastic solids.

We conclude this section by emphasizing that the present formulation of Relativistic Hyperelasticity does not require the definition of a time function (which is indeed not a necessity in astrophysics) and the associated assumption of a foliation of the World tube $\mW$ by spacelike hypersurfaces. All we need is to endow the body $\body$ with a fixed metric $\bgamma_{0}$ as in~\cite{Sou1958,Sou1964,CQ1972}.

% ----------------------------------------------------------------
\section{The stress--energy tensor}
\label{sec:stress-energy-tensor}
% ----------------------------------------------------------------

The \emph{stress-energy tensor}, also called \emph{energy-momentum tensor} can be considered as a four-dimensional generalization of the stress tensor in Classical three-dimensional Continuum Mechanics. In General Relativity, it is the source of the curvature of the metric $g$ of the Universe. It is usually defined as the variational derivative of a Lagrangian with respect to the metric $g$ \cite{Hil1915,Hil1924,Noe1918,Kos2011,BGRS2016} and, for this reason, it is thus a symmetric contravariant second-order tensor field (or a \emph{tensor distribution} defined on symmetric second-order covariant tensor fields, in more general situations \cite{Sou1997a}).

In the present case, the Euler-Lagrange stationary equation $\delta \mL=0$ for the Lagrangian
\begin{equation*}
  \mL(g,\Psi)= \mH(g) + \mL^{\text{matter}}(g, \Psi),
\end{equation*}
leads in particular to the equation
\begin{equation*}
  \frac{\delta \mH}{\delta g} + \frac{\delta \mL^{\text{matter}}}{\delta g} = 0,
\end{equation*}
when only variations of the metric $g$ are considered. It recasts as the Einstein field equation
\begin{equation}\label{eq:Einstein}
  \Ein_{g}^{\sharp} + \Lambda g^{-1} = \kappa\bT,
\end{equation}
if $\Ein_{g}^{\sharp}=g^{-1}\Ein_{g} g^{-1}$ is the contravariant form of Einstein's tensor~\eqref{eq:Einstein-tensor}, and
\begin{equation*}
  \bT:=-2\frac{\delta \mL^{\text{matter}}}{\delta g}.
\end{equation*}
is the \emph{stress-energy tensor} (the source term in Einstein's equation), which is a symmetric contravariant
second-order tensor field on the Universe $\mM$.

\begin{rem}
  Because $\dive^{g}(\Ein_{g}^{\sharp} + \Lambda g^{-1})=0$ (see remark~\ref{rem:H-invariance}), the stress-energy tensor $\bT$ satisfies the conservation law
  \begin{equation*}
    \dive^{g}\bT=0.
  \end{equation*}
  As observed by Einstein himself~\cite{Ein1988}, ``$\dive \bT = 0$, that's mechanics''. Indeed, this equation generalizes in 4D (and non flat Universe) the three-dimensional equilibrium equations of Classical Continuum Mechanics. When a spacetime structure is adopted, the Cauchy stress tensor is related to the spacelike components of $\bT$ (see \autoref{sec:hyperelasticity-spacetime}).
\end{rem}

The following result provides a general expression for the stress-energy tensor of $\bT$ in the case of Relativistic Hyperelasticity (see also \cite{KM1992}).

\begin{thm}[Souriau, 1958]\label{thm:relativistic-hyperelasticity}
  Consider the general covariant matter Lagrangian
  \begin{equation*}
    \mL^{\text{matter}}(g, \Psi) = \int L\, \vol_{g},
    \qquad
    L=\rho_{r} c^{2}+E,
  \end{equation*}
  with
  \begin{equation*}
    \rho_{r} = \rho_{\bgamma_{0}}(\Psi) \sqrt{\det\left[ \bH(\bgamma_{0} \circ \Psi)\right]} \quad \text{and} \quad E = E(\Psi, \bH),
  \end{equation*}
  and where $\bgamma_{0}$ is a fixed metric on the body $\body$. Then, its stress-energy tensor has the following expression
  \begin{equation}\label{eq:def-S}
    \bT = -2\frac{\delta \mL^{\text{matter}}}{\delta g} = \rho_{r} c^{2} \bU \otimes \bU - \bS,
  \end{equation}
  where
  \begin{equation*}
    \bS := E\,g^{-1} -2 g^{-1} (T\Psi)^{\star}\frac{\partial E}{\partial \bH} (T\Psi) g^{-1}.
  \end{equation*}
  is the (four-dimensional) \emph{relativistic stress tensor} on $\mW$. Moreover, we have
  \begin{equation*}
    \bS \cdot \bU^{\flat}=E\, \bU,
    \quad
    \text{and}
    \quad
    \bT \cdot \bU^{\flat}=-L \bU.
  \end{equation*}
\end{thm}

\begin{rem}[Bennoun, 1965]\label{rem:T}
  Since $ \bS \cdot \bU^{\flat}\neq 0$, the decomposition~\eqref{eq:def-S} is not an orthogonal decomposition relative to $\bU$ (see \autoref{sec:Orth-Decomp-2nd}). Writing $E = \rho_{r} e$, with $e = e(\Psi, \bH)$, the specific internal energy, the stress-energy tensor naturally recasts, using its orthogonal decomposition relative to $\bU$, as
  \begin{equation}\label{eq:def-Sigma}
    \bT = L\,\bU \otimes \bU -\bSigma,
  \end{equation}
  where its spatial part
  \begin{equation}\label{eq:dedH}
    \bSigma:=- 2\rho_{r}\, g^{-1} (T\Psi)^{\star}\frac{\partial  e}{\partial \bH} (T\Psi) g^{-1},
  \end{equation}
  is such that
  \begin{equation*}
    \bSigma=\bS + E\, \bU \otimes \bU
    \quad \text{and} \quad \bSigma \cdot \bU^{\flat}=0,
  \end{equation*}
  can also be interpreted as a (four-dimensional) relativistic stress tensor.
\end{rem}

\begin{proof}
  Consider the variation $\delta_{g}\mL^{\text{matter}}$ of the Lagrangian $\mL^{\text{matter}}$ with respect to the metric $g$. Then, we have
  \begin{equation*}
    \delta_{g}\mL^{\text{matter}} = \int (\delta_{g}L)\, \vol_{g} + L\, \delta_{g}\vol_{g},
  \end{equation*}
  with
  \begin{equation*}
    \delta_{g}L = \tr \left(\frac{\partial L}{\partial \bH} \delta_{g}\bH\right) \quad \text{and} \quad \delta_{g} \vol_{g} = \frac{1}{2} \tr(g^{-1}\delta g) \vol_{g}.
  \end{equation*}
  But
  \begin{equation*}
    \delta_{g}\bH = - (T\Psi) g^{-1}\delta g\, g^{-1}(T\Psi)^{\star},
  \end{equation*}
  and hence
  \begin{equation*}
    \delta_{g}L  = - \tr \left(\frac{\partial L}{\partial \bH} (T\Psi) g^{-1}\delta g\, g^{-1}(T\Psi)^{\star}\right)
    = - \tr \left(g^{-1}(T\Psi)^{\star}\frac{\partial L}{\partial \bH} (T\Psi) g^{-1}\delta g \right).
  \end{equation*}
  We get thus
  \begin{equation*}
    \delta_{g}\mL^{\text{matter}} = - \int \tr \left[ \left( g^{-1}(T\Psi)^{\star}\frac{\partial L}{\partial \bH} (T\Psi) g^{-1} - \frac{1}{2} L g^{-1} \right) \delta g \right]\vol_{g},
  \end{equation*}
  and therefore
  \begin{equation*}
    \bT = 2 g^{-1} (T\Psi)^{\star}\frac{\partial L}{\partial \bH} T\Psi g^{-1}-  L\,g^{-1}.
  \end{equation*}
  Now, we have
  \begin{equation*}
    \frac{\partial \rho_{r}}{\partial \bH} = \frac{1}{2} \rho_{r} \; \bH^{-1},
  \end{equation*}
  and thus
  \begin{align*}
    \bT & = \rho_{r}c^{2}  \, g^{-1} (T\Psi)^{\star}\bH^{-1} (T\Psi) \, g^{-1}
    +2 g^{-1} (T\Psi)^{\star}\frac{\partial E}{\partial \bH} (T\Psi) g^{-1}
    -  \rho_{r} c^{2} \,g^{-1}-  E\,g^{-1}
    \\
        & = \rho_{r} c^{2}\left[g^{-1}h g^{-1}-   \,g^{-1}\right]
    +2 g^{-1} (T\Psi)^{\star}\frac{\partial E}{\partial \bH} (T\Psi) g^{-1}
    -  E\,g^{-1}
    \\
        & =\rho_{r} c^{2} \bU \otimes \bU+
    2 g^{-1} (T\Psi)^{\star}\frac{\partial E}{\partial \bH} (T\Psi) g^{-1}
    -  E\,g^{-1},
  \end{align*}
  where the second equality is gained by lemma~\ref{lem:h-H}, and the third one by~\eqref{eq:def-h}. Finally, since $T\Psi.\bU=0$ and $\bU$ is a unit timelike vector field, we get $ \bS \cdot \bU^{\flat}=E\, \bU$, and $\bT \cdot \bU^{\flat}=-L \bU$.
\end{proof}

\begin{exam}[Relativistic perfect fluid]\label{rem:RF}
  The stress-energy tensor of a Relativistic perfect fluid,
  \begin{equation*}
    \bT = \left(L  + P \right)\bU \otimes \bU + P\, g^{-1},
    \qquad
    L=\rho_{r} c^{2}+ E,
  \end{equation*}
  corresponds to an internal energy density of the form $E= \rho_{r}e(\rho_{r})$, where $P = \rho_{r}^{2}  e'(\rho_{r})$ is the pressure. Indeed, in that case, we have
  \begin{equation*}
    \frac{\partial e}{\partial \bH} = e^{\prime}(\rho_{r})\frac{\partial \rho_{r}}{\partial \bH} \quad \text{with} \quad  \frac{\partial \rho_{r}}{\partial \bH} = \frac{1}{2} \rho_{r} \; \bH^{-1},
  \end{equation*}
  and thus, by lemma~\ref{lem:h-H}, we get
  \begin{equation*}
    (T\Psi)^{\star}\frac{\partial  e}{\partial \bH} (T\Psi) = \frac{1}{2}\rho_{r}e^{\prime}(\rho_{r}) (T\Psi)^{\star}\bH^{-1}(T\Psi) = \frac{1}{2}\rho_{r}e^{\prime}(\rho_{r}) h.
  \end{equation*}
  Therefore
  \begin{equation*}
    \bSigma = - 2\rho_{r}\, g^{-1} (T\Psi)^{\star}\frac{\partial  e}{\partial \bH} (T\Psi) g^{-1} = - \rho_{r}^{2}  e'(\rho_{r})\, g^{-1} h g^{-1} =  - P (g^{-1} + \bU\otimes\bU),
  \end{equation*}
  where we have set $P := \rho_{r}^{2}  e'(\rho_{r})$, and we get
  \begin{equation*}
    \bS = \bSigma - E\,\bU \otimes \bU = -(E + P)\bU\otimes\bU - P \,g^{-1}.
  \end{equation*}
  The corresponding stress--energy tensor is thus given by
  \begin{equation*}
    \bT = \rho_{r} c^{2} \bU \otimes \bU - \bS = L \bU \otimes \bU - \bSigma = (\rho_{r} c^{2} + E + P)\bU\otimes\bU + P \,g^{-1}.
  \end{equation*}
\end{exam}

Even if the full theory is four-dimensional, the orthogonal decomposition~\eqref{eq:def-Sigma} of $\bT$ relative to $\bU$, and the definition~\eqref{eq:dedH} (\emph{i.e.}, the Relativistic Hyperelasticity law) naturally introduce a three-dimensional symmetric stress tensor, either \emph{covariant},
\begin{equation}\label{eq:def-s-3D}
  \bs:=- 2\frac{\partial  e}{\partial \bH} ,
\end{equation}
or, \emph{contravariant},
\begin{equation*}
  \bH\,\bs\,\bH=2\frac{\partial  e\;\,}{\partial \bH^{-1}},
\end{equation*}
since the conformation is invertible (and contravariant). The stress tensors $\bs$ and $\bH\,\bs\,\bH$ are generalizations of the \emph{second Piola--Kirchhoff stress tensor} (expressed on a reference configuration $\Omega_{0}$ of Classical Continuum Mechanics) or more precisely here of the \emph{Rougée stress tensor} \cite{Rou1991a,Rou2006} (defined on the body $\body$, see~\autoref{sec:stress-on-the-body}).
These constitutive equations are the three-dimensional Relativistic Hyperelasticity laws. They do not depend on the further assumption of a foliation of the World tube $\mW$, nor on the consideration of a spacetime.

The underlying question~\cite{EBT2006,Bro2021} is then how to properly import in General Relativity existing Classical Continuum Mechanics constitutive laws formulated on the body $\body$ \cite{Rou2006,KD2021} (or a reference configuration $\Omega_{0}$). Indeed, many three-dimensional expressions of energy densities
\begin{equation}\label{eq:w3D}
  w=w(\widehat{\bgamma}), \quad \widehat{\bgamma}:=\bgamma_{0}^{-1}\bgamma, \quad \text{on $\body$}
  \qquad \Big( \emph{i.e.}, \;w=w(\widehat\bC), \; \widehat\bC:=q^{-1}\bC, \; \text{when $\body \equiv \Omega_{0}$} \Big),
\end{equation}
are available in the Classical Continuum Mechanics literature \cite{Moo1940,Har1966,Ogd1972,AB1993,Sto2009,GMDC2011}.
They are local function of the mixed tensor $\widehat{\bgamma} = \bgamma_{0}^{-1}\bgamma$, defined using the reference metric $\bgamma_{0}$ on $\body$ (equivalently, of the mixed right Cauchy--Green tensor $\widehat\bC$ on $\Omega_{0}$), meaning that
\begin{equation*}
  w(\bX)=w\left(\widehat\bgamma(\bX) \right), \quad \bX \in \body.
\end{equation*}

The use of such energy densities is then straightforward in the Relativistic framework, if one sets (using definition~\eqref{eq:H0})
\begin{equation}\label{eq:e-from-w}
  e(\Psi, \bH)=w\left((\bgamma_{0}^{-1}\circ \Psi\right)\bH^{-1})=w\left(\bH_{0}\,\bH^{-1}\right).
\end{equation}
Using~\eqref{eq:e-from-w}, we get
\begin{equation}\label{eq:s-3D}
  \bs = -2\frac{\partial e}{\partial \bH}=2\,\bH^{-1} \bH_{0}\frac{\partial  w}{\partial \widehat{\bgamma}}\bH^{-1},
  \qquad
  \bH\, \bs \, \bH=2\, \bH_{0}\frac{\partial  w}{\partial \widehat{\bgamma}},
\end{equation}
so that the stress-energy tensor $\bT$ and the four-dimensional stress $\bSigma$ recast finally  as
\begin{equation}\label{eqTs-3D}
  \bT= L \bU \otimes \bU-\bSigma, \qquad \bSigma=\rho_{r}\, g^{-1} (T\Psi)^{\star}\, \bs\, (T\Psi) g^{-1}.
\end{equation}
We refer to \autoref{sec:stress-on-the-body}  for the full link ---which needs the consideration of a spacetime--- with stresses on the body $\body$.

% ----------------------------------------------------------------
\section{Universe's foliation by spacelike hypersurfaces}
\label{sec:Universe-foliation}
% ----------------------------------------------------------------

There is no Mechanics without the proper definition of \emph{time} and \emph{space}. To introduce these concepts in General Relativity, one usually starts by introducing a smooth submersion (a \emph{time function}) $\hat{t}$ on the Universe $\mM$ with a timelike gradient everywhere. Then, spacelike hypersurfaces are defined as
\begin{equation}\label{eq:Omegatini}
  \espace_{t} := \set{m \in \mM;\; \hat{t}(m) = t},
\end{equation}
and one expects the Universe to be foliated by these hypersurfaces~\cite{Lic1955,Gou2012}. The problem is that, in general, a global foliation of the Universe might not exist (see~\cite[Chapter 4]{Gou2012}). Anyway, if such a foliation exists or is given, one will say that the Universe has been endowed with a \emph{spacetime structure} or a \emph{$(3+1)$-structure} as defined in~\cite{Dar1927,Lic1939,FBCB1948,Lic1952,Fou1956,ADM1962,Yor1979,Gou2012}.

Fortunately, for our concerns, we do not have to address this problem \emph{globally}. In the present paper, we will simply admit that such a foliation exists on a \emph{local chart} which contains the body World tube $\mW$, or a part of it. Indeed, in Mechanical Engineering, the spacetime domain occupied by a continuous medium/a structure, embedded for example in a laboratory, a building, a city, a country, a domain of space \ldots can be considered as included into such a local chart.

Moreover, since the presence of the studied matter in the laboratory does not affect (much) the Universe metric $g$ compared to the one of Earth (\emph{passive matter} assumption), we can choose, among the numerous spacetimes encountered in General Relativity and available in~\cite{HE1973,MTW1973,MG2009}, those describing solutions of Einstein equations in the vacuum. These spacetimes are usually described using a coordinate system $(x^{\mu})$, for which the time function is chosen as
\begin{equation*}
  \hat{t}=\frac{x^{0}}{c},
\end{equation*}
and we then define $\Omega_{t}$ as the intersection of the World tube $\mW$ with the spacelike hypersurface $\espace_{t}$,
\begin{equation*}
  \Omega_{t} := \mW \cap \espace_{t} = \mW \cap \set{x^{0}=ct}.
\end{equation*}
The three-dimensional hypersurfaces $\Omega_{t}$ of the Universe play the same role as the configurations, parameterized by time $t$, of Classical Continuum Mechanics~\cite{TN1965,Nol1972,Nol1974,MH1994,Ber2012,Ste2015,For2022}, with the difference that the later are embedded in the three-dimensional Euclidean space $\espace$, not in the four-dimensional Universe $\mM$. This construction is illustrated in \autoref{fig:foliationOmegat}, where a second time, $t_{0}$, and the associated hypersurface $\Omega_{t_{0}}$ (a possible reference configuration) are represented.

\begin{figure}[ht]
  \centering
  \includegraphics[width=12cm]{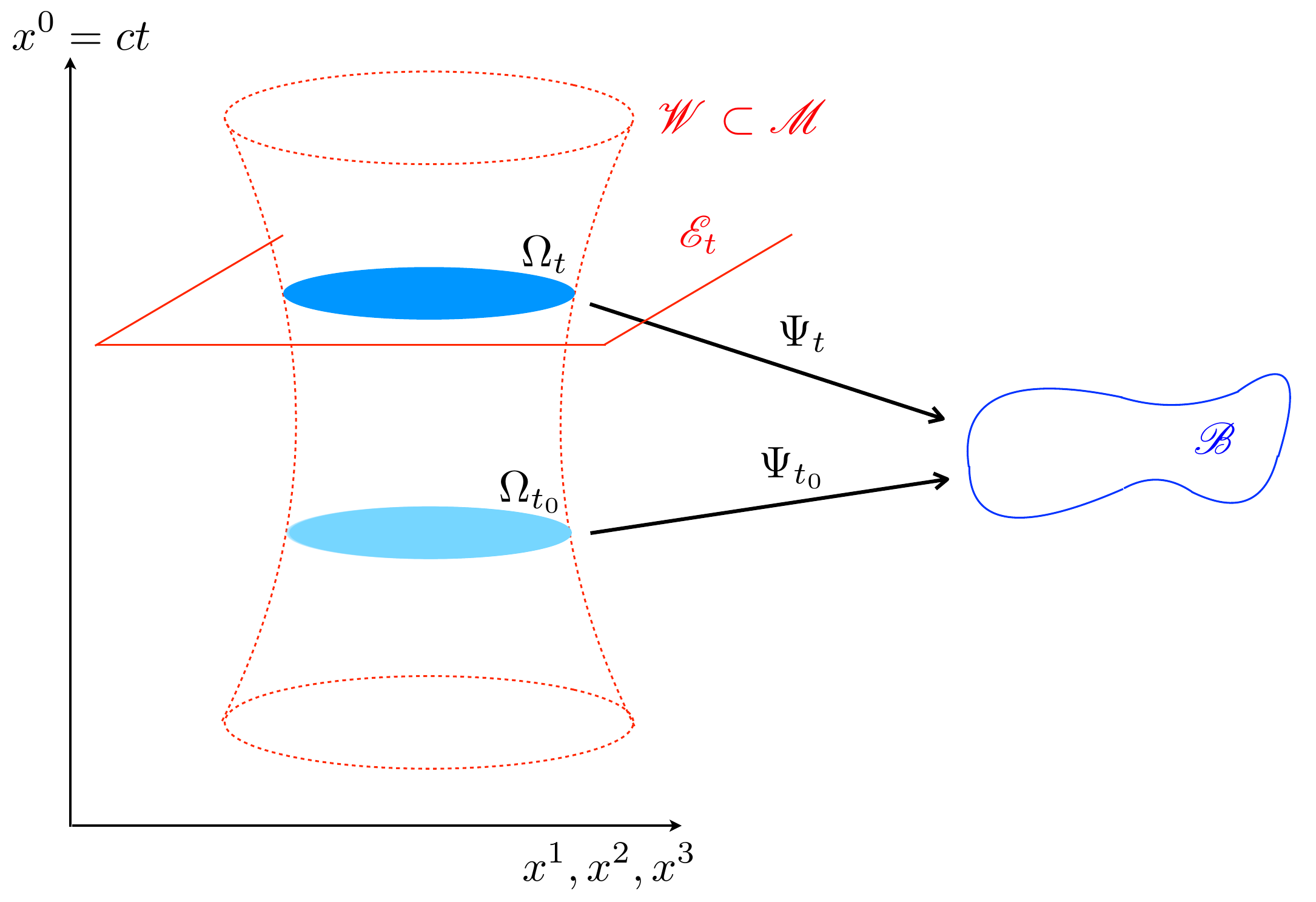}
  \caption{The foliation of the World tube $\mW$ by spacelike hypersurfaces $\Omega_{t}$.}
  \label{fig:foliationOmegat}
\end{figure}

The canonical embedding of these submanifolds into the Universe $\mM$ is noted $j_{t}\colon \Omega_{t} \to \mM$ (rather than $\theta$ as in \cite[Chapter 7]{HE1973}). Then, for each $t$, the pullback of the matter field $\Psi$ by~$j_{t}$,
\begin{equation*}
  \Psi_{t} := j_{t}^{*} \Psi = \Psi \circ j_{t} \colon \Omega_{t} \to V,
\end{equation*}
is just the restriction of $\Psi$ to $\Omega_{t}$, and we have $\Psi_{t}(\Omega_{t})=\body$.

\begin{rem}\label{rem:TPsit-iso}
  We have not made, so far, the assumption that $\Psi_{t}$ is a diffeomorphism. However, the tangent map
  \begin{equation*}
    T_{\xx} \Psi_{t} \colon T_{\xx} \Omega_{t} \to T_{\Psi_{t}(\xx)}\body
  \end{equation*}
  is an isomorphism for each $\xx \in \Omega_{t}$ and each $t$, since we have assumed that $\Psi$ is a submersion. We shall denote the inverse mapping of $T\Psi_{t}$ by $\bF(t)$ and omit, when not necessary, the explicit dependence on time and write simply $\bF = (T\Psi_{t})^{-1}$. If we make, furthermore, \emph{the stronger hypothesis} that $\Psi_{t}$ is a diffeomorphism we can set
  \begin{equation*}
    \pp(t)= \Psi_{t}^{-1}\colon \body \to \Omega_{t},
  \end{equation*}
  and then $\bF=T\pp$.
\end{rem}

The unit normal to the spacelike hypersurfaces $\Omega_{t}$, proportional to the gradient of the time function $\hat{t}$, is denoted by the quadrivector $\bN$. We have two opposite choices to define such a unit vector and we define $\bN$ as~\cite{Gou2012}
\begin{equation}\label{eq:defN}
  \bN := -\frac{\grad\, \hat{t}}{\sqrt{-\norm{\grad\, \hat{t}}^{2}}},
\end{equation}
where the gradient and the norm are relative to the metric $g$. The minus sign is chosen so that the quadrivector $\bN$ is \emph{future-oriented}, meaning that the value of $\hat{t}$ increases along the flow curves of $\bN$.
Note that, at each point $m \in \mM$ where the spacetime structure is defined, we have the orthogonal decomposition
\begin{equation*}
  T_{m}\mM = \langle \bN(m) \rangle \oplus \bN(m)^{\bot},
\end{equation*}
where, for each $m\in \Omega_{t}$, the orthogonal subspace $\bN(m)^{\bot} = T_{m}\Omega_{t}$ coincides with the tangent space at $m$ to the spacelike hypersurface $\Omega_{t}$.

An important special case, and the only one used in this paper for our description of Relativistic Continuum Mechanics of solids, is the one of a \emph{static spacetime}. Such a spacetime is induced by a \emph{static metric}, \emph{i.e.}, a metric $g$ for which there exists a timelike Killing vector field $X$ (\textit{i.e.} $\Lie_{X}g=0$), which is moreover the gradient of a time function $\hat{t}$. There exists then a coordinate system $(x^{\mu})$, with $x^{0}=ct$ the time coordinate, for which
\begin{equation*}
  \frac{\partial g_{\mu \nu}}{\partial x^{0}} = 0 \quad \text{and} \quad g_{0i}=0.
\end{equation*}
In that case, we get
\begin{equation*}
  \grad\, \hat{t} = \frac{1}{c} g^{00}\partial_{x^{0}},
\end{equation*}
and the unit normal $\bN$ is written as
\begin{equation*}
  \bN = \sqrt{-g^{00}}\partial_{x^{0}}.
\end{equation*}
Examples of static spacetimes are the Minkowski and the Schwarzschild~\cite{Sch1916} spacetimes.

% ----------------------------------------------------------------
\section{Matter field in a spacetime -- Generalized Lorentz factor}
\label{sec:Matter-field-spacetime}
% ----------------------------------------------------------------

Perfect matter in the Universe $\mM$ is represented by a matter field $\Psi$. This field generates a timelike quadrivector $\bP$ on the World tube $\mW$, as introduced in \autoref{sec:matter-fields}, and a unit timelike quadrivector $\bU=\bP/\rho_{r}$. Therefore, if a spacetime structure is introduced on the Universe as it has been explained in \autoref{sec:Universe-foliation}, we get a second unit timelike quadrivector $\bN$, normal to the hypersurfaces $\Omega_{t}$, and in general not collinear to $\bU$. By changing the sign of the time function $\hat{t}$ if necessary, we can assume, anyway, that both $\bU$ and $\bN$ define the same time orientation. This is characterized by the condition
\begin{equation*}
  \langle \bU, \bN \rangle_{g} < 0.
\end{equation*}
Thus, the quadrivector $\bU$ can be written uniquely using the orthogonal decomposition
\begin{equation*}
  \bU = \bU^{N} + \bU^{\top},
\end{equation*}
where $\bU^{N} = -\langle \bU, \bN \rangle_{g}\,\bN$ is the normal component of $\bU$ and $\bU^{\top}$ is the tangential component to $\Omega_{t}$. Introducing the function (see for instance~\cite{Gou2006})
\begin{equation}\label{eq:def-gamma}
  \gamma: = - \langle \bU, \bN\rangle_{g},
\end{equation}
one can write thus,
\begin{equation}\label{eq:U-orthogonal-decomposition}
  \bU^{N} = \gamma \bN, \quad \text{and} \quad \bU^{\top} = \frac{\gamma}{c} \uu,
\end{equation}
where $c$ is the light speed, and $\uu \in \bN^{\bot} = T\Omega_{t}$.

\begin{rem}
  Since we deal with a foliation by hypersurfaces $\Omega_{t}$, rather than just one hypersurface, the tangential component $\bU^{\top}$ of a vector field $\bU$ defined on $\mM$ (or an open subset of $\mM$) can be simultaneously interpreted as a vector field defined on $\mM$ but tangential to $\Omega_{t}$ at each point $m$, or as a (time-dependent) vector field on $\Omega_{t}$ (when restricted to $\Omega_{t}$). In our notations, we do not distinguish between these two interpretations.
\end{rem}

The orthogonal decomposition $\gamma \bN + \frac{\gamma}{c} \uu$ of $\bU$ and the relation $T\Psi. \bU=0$ deduced from~\eqref{eq:TPsi-P-null} allow to express the three-dimensional velocity $\uu$ on $\Omega_{t}$, as
\begin{equation}\label{eq:def-u}
  \uu = -c\, \bF\, T\Psi. \bN,
  \qquad
  \bF^{-1} := T\Psi_{t},
\end{equation}
where $\bF\colon T\body \to T \Omega_{t}$ is defined as the inverse of $T\Psi_{t}$, the restriction of $T\Psi$ to $T\Omega_{t}$, which is an invertible linear mapping (remark~\ref{rem:TPsit-iso}). In \autoref{sec:Gallilean}, the expression~\eqref{eq:def-u} will allow us to interpret the Galilean limit of $\uu$ as the three-dimensional Eulerian velocity on $\Omega_{t}$.

Using the fact that $\norm{\bU}^{2}_{g}=-1$, one gets furthermore that
\begin{equation*}
  \gamma = \frac{1}{\sqrt{1-\displaystyle\frac{\norm{\uu}_{g}^{2}}{c^{2}}}},
\end{equation*}
and $\gamma\geq 1$ since $\uu$ is spacelike. This function plays a fundamental role in General Relativity and its notation is not accidental. In the special case of the Minkowski spacetime, where $g=\eta$ is the Minkowski metric, one recovers the traditional Lorentz factor
\begin{equation*}
  \gamma = \frac{1}{\sqrt{1-\frac{\uu^{2}}{c^{2}}}},
\end{equation*}
where $\uu^{2}:=\norm{\uu}_{q}^{2}$ is the square Euclidean norm of $\uu$. For this reason, we shall call $\gamma = - \langle \bU, \bN\rangle_{g}$ the \emph{generalized Lorentz factor}.

\begin{rem}[Rest frame and observers]
  The concept of \emph{rest frame} is well-defined for a particle in Special Relativity. Its definition for distributed matter in general Relativity is much less clear. In this paper, we will adopt the following definition. Given a matter field $\Psi$, a \emph{rest frame} will be defined as a spacetime in which $\bU = \bN$, \emph{i.e.}, as a spacetime in which the generalized Lorentz factor is $\gamma=1$. For such a spacetime, we will get of course $\uu = 0$ and the particles can be considered at rest in it. The corresponding time coordinate will thus be interpreted as the \emph{proper time} $\tau$. More generally and heuristically, we can interpret $\bU$ as ``inducing a splitting of the tangent bundle $T\mW$ for matter'' and $\bN$ as ``inducing an integrable splitting or (3+1) spacetime for an observer''. The Lorentz factor $\gamma = -\langle \bU, \bN\rangle_{g}$ is then the ``angle'' between these two timelike directions.
\end{rem}

Finally, the normal component of the current of matter $\bP=\rho_{r} \bU$ (definition~\eqref{eq:def-U}) is then simply
\begin{equation*}
  \bP^{N}=\gamma \rho_{r}\, \bN= \rho\, \bN,
\end{equation*}
where $\bN$ is the unit timelike normal to the hypersurfaces $\Omega_{t}$.
The function
\begin{equation}\label{eq:rho}
  \rho:= \gamma \rho_{r},
\end{equation}
defined on the World tube $\mW$, is interpreted as the \emph{relativistic mass density}. A geometric interpretation of the restriction of $\rho$ to $\Omega_{t}$ is provided in \autoref{sec:3D-metrics}.

% ----------------------------------------------------------------
\section{Relativistic stress tensors and constitutive laws in a spacetime}
\label{sec:hyperelasticity-spacetime}
% ----------------------------------------------------------------

We assume here that the World tube $\mW$ is foliated by three-dimensional hypersurfaces $\Omega_{t}$, with unit timelike normal $\bN$. Then, it is possible to use the orthogonal decomposition of each tangent space relative to $\bN$ to split each tensor field accordingly. These splittings are referred to as \emph{(3+1)-decompositions} in the General Relativity literature~\cite{ADM1962,Yor1979,Gou2012}. Explicit formulas for second-order tensors are provided in \autoref{sec:Orth-Decomp-2nd}. We follow here the calculations of Souriau~\cite{Sou1958,Sou1964} for the flat Minkowski spacetime and extend them to any spacetime, thanks to this (3+1)-decomposition. These calculations generalize the ones given for relativistic fluids in~\cite{Gou2012} to relativistic solids. In particular, the orthogonal decomposition of the stress-energy tensor allows us to introduce relativistic generalizations of the \emph{Cauchy stress tensor} as 3D tensors and to recast the 4D Relativistic Hyperelasticity law (theorem~\ref{thm:relativistic-hyperelasticity}) as a \emph{three-dimensional constitutive equation}, relating these 3D stress tensors to the conformation $\bH$.

The orthogonal decompositions of $g$ and $g^{-1}$, relative to the unit timelike vector $\bN$ (instead of $\bU$ as in \autoref{sec:Conformation-strains}) are
\begin{equation}\label{eq:g3D}
  g = \mg - \bN^{\flat} \otimes \bN^{\flat},
  \qquad
  g^{-1} = (\mg)^{\sharp} - \bN \otimes \bN,
  \qquad
  \mg \bN = 0,
\end{equation}
where the degenerate quadratic form $\mg$ and $(\mg)^{\sharp}=g^{-1} \mg g^{-1}$ are of signature $(0, +, +, +)$, by lemma~\ref{lem:h-signature}. This decomposition allows, in particular, to recast the conformation as
\begin{equation*}
  \bH = T\Psi \, g^{-1} \, (T\Psi)^{\star} = T\Psi ((\mg)^{\sharp} - \bN \otimes \bN) (T\Psi)^{\star},
\end{equation*}
with
\begin{equation*}
  T\Psi (\bN \otimes \bN) (T\Psi)^{\star} = T\Psi \left(\frac{1}{\gamma} \bU-\frac{1}{c} \uu \right) \otimes \left(\frac{1}{\gamma} \bU-\frac{1}{c} \uu\right) (T\Psi)^{\star} = \bF^{-1} \left(\frac{1}{c^{2}} \uu \otimes \uu\right) \bF^{-\star},
\end{equation*}
since $T\Psi. \bU=0$, and thus
\begin{equation}\label{eq:conformation-in-spacetime}
  \bH = \bF^{-1}\left( (\mg)^{\sharp} -\frac{1}{c^{2}} \uu \otimes \uu\right) \bF^{-\star}.
\end{equation}

When applied to the stress-energy tensor $\bT$, the orthogonal decomposition~\eqref{eq:T-orthogonal-decomposition} leads to
\begin{equation}\label{eq:stress-energy-tensor-decomposition}
  \bT = E_{\mathrm{tot}} \, \bN \otimes \bN + \frac{1}{c} \left(\bN \otimes \bp + \bp \otimes \bN \right) + \bt,
\end{equation}
and allows to define the physical components of $\bT$ in the considered spacetime:
\begin{itemize}
  \item $E_{\mathrm{tot}}$ is the \emph{total energy density},
  \item $\bp$ is the \emph{momentum density vector field},
  \item and the spatial part $\bt$ of $\bT$ is related to the \emph{stress field}.
\end{itemize}
The question asked by Souriau is then: How are these quantities related to a three-dimensional relativistic generalization of the \emph{Cauchy stress tensor} $\bsigma$? As shown below, the answer depends on the choice of the decomposition  of the stress-energy tensor (see theorem~\ref{thm:relativistic-hyperelasticity} and remark~\ref{rem:T}). Indeed, we have seen  that there are two possible splittings of it:
\begin{enumerate}
  \item $\bT = \rho_{r}c^{2} \bU \otimes \bU - \bS$ (considered by Souriau \cite{Sou1958,Sou1964} and Synge \cite{Syn1959}),
  \item or, $\bT = L \bU \otimes \bU - \bSigma$ (considered by Eckart \cite{Eck1940} and Bennoun \cite{Ben1965}).
\end{enumerate}
The orthogonal decompositions (relative to $\bN$) of the two four-dimensional stresses $\bS$ and $\bSigma$ interestingly give rise to two possible ways to define a three-dimensional stress tensor $\bsigma$:
\begin{enumerate}
  \item either as the spatial part of $\bS$,
  \item or, as the spatial part of $\bSigma$.
\end{enumerate}

\subsection*{First choice: $\pmb \sigma$ is defined as the spatial part of $\bS$}

Using the fact that $\bS\cdot\bU^{\flat}=E\, \bU$ by theorem~\ref{thm:relativistic-hyperelasticity} with $\bU^{\flat} = \gamma \bN^{\flat} + \gamma \uu^{\flat}/ c$ by~\eqref{eq:U-orthogonal-decomposition}, the orthogonal decomposition of $\bS$ can be expressed as
\begin{equation}\label{eq:S-from-sigma}
  \bS = \alpha \, \bN \otimes \bN + \bN \otimes \aaa +\aaa \otimes \bN + \bsigma,
\end{equation}
where the spatial part of $\bS$ has been set equal to $\bsigma$, and
\begin{equation*}
  \aaa = \frac{1}{c}\left(\bsigma \cdot \uu^{\flat} - E\uu\right), \quad \text{and} \quad \alpha =  \frac{1}{c^{2}}\uu^{\flat} \cdot \bsigma \cdot \uu^{\flat} - E\left( 1 + \frac{\norm{\uu}^{2}}{c^{2}}\right).
\end{equation*}
The associated orthogonal decomposition of the stress-energy tensor $\bT = \rho_{r}c^{2} \bU \otimes \bU - \bS$ is then
\begin{equation*}
  \begin{cases}
    E_{\mathrm{tot}} = \gamma \rho c^{2} + E\Big(1+\frac{1}{c^{2}}  \norm{\uu}^{2}\Big) -\frac{1}{c^{2}} \uu^{\flat} \cdot \bsigma \cdot \uu^{\flat},
    \\
    \bp =(\gamma \rho c^{2}+E)\uu - \bsigma \cdot \uu^{\flat},
    \\
    \bt = \gamma\rho \, \uu \otimes \uu - \bsigma,
  \end{cases}
\end{equation*}
where $\rho= \gamma \rho_{r}$ is the relativistic mass density. The three-dimensional stress tensor $\bsigma$, defined as the spatial part of
\begin{equation}\label{eq:Hyperelasticity-law-S}
  \bS = - 2\rho_{r}\, g^{-1} (T\Psi)^{\star}\frac{\partial e}{\partial \bH} (T\Psi) g^{-1}- E\, \bU \otimes \bU,
\end{equation}
is thus
\begin{equation}\label{eq:Hyperelasticity-a}
  \bsigma = - \frac{2}{\gamma} \rho\, (\mg)^{\sharp} \bF^{-\star}\frac{\partial e}{\partial \bH} \bF^{-1} (\mg)^{\sharp} - \frac{\gamma^{2}E}{c^{2}} \, \uu \otimes \uu.
\end{equation}
The later equation can be interpreted as a three-dimensional Hyperelasticity law. Introducing the generalized second Piola--Kirchhoff stress tensor~\eqref{eq:def-s-3D}, we get
\begin{equation*}
  \bsigma = \frac{1}{\gamma} \rho\, (\mg)^{\sharp} \bF^{-\star}\,\bs\,\bF^{-1} (\mg)^{\sharp} - \frac{\gamma^{2}E}{c^{2}} \, \uu \otimes \uu,
  \qquad
  \bs=\displaystyle-2\frac{\partial e}{\partial \bH}.
\end{equation*}

\subsection*{Second choice: $\pmb \sigma$ is defined as the spatial part of $\pmb\Sigma$}

Using this time the fact that $\bSigma \cdot \bU^{\flat}=0$ by remark~\ref{rem:T}, we get the orthogonal decomposition
\begin{equation}\label{eq:Sigma-from-sigma}
  \bSigma =  \frac{1}{c^{2}}(\uu^{\flat} \cdot \bsigma \cdot \uu^{\flat}) \, \bN \otimes \bN
  + \frac{1}{c} \left(\bN \otimes ( \bsigma \cdot \uu^{\flat}) +( \bsigma \cdot \uu^{\flat}) \otimes \bN \right)
  + \bsigma,
\end{equation}
where the spatial part of $\bSigma$ has been set equal to $\bsigma$. The associated orthogonal decomposition of the stress-energy tensor $\bT = L \bU \otimes \bU - \bSigma$, with $L=\rho_{r}c^{2}  + E$, is now
\begin{equation*}
  \begin{cases}
    E_{\mathrm{tot}} =  \gamma^{2} L -\frac{1}{c^{2}} \uu^{\flat} \cdot \bsigma \cdot \uu^{\flat}
    ,
    \\
    \bp = \gamma^{2} L\, \uu - \bsigma \cdot \uu^{\flat},
    \\
    \bt = \frac{\gamma^{2}}{c^{2}} L\, \uu \otimes \uu - \bsigma ,
  \end{cases}
\end{equation*}
where, using $\gamma^{2}=1+\frac{\gamma^{2}}{c^{2}} \norm{\uu}^{2}$,
\begin{equation*}
  \gamma^{2} L = \gamma \rho c^{2} + E\Big(1+\frac{\gamma^{2}}{c^{2}} \norm{\uu}^{2}\Big) .
\end{equation*}
The three-dimensional stress tensor $\bsigma$, defined as the spatial part of
\begin{equation}\label{eq:Hyperelasticity-law-Sigma}
  \bSigma = - 2\rho_{r}\, g^{-1} (T\Psi)^{\star}\frac{\partial e}{\partial \bH} (T\Psi) g^{-1}.
\end{equation}
is then
\begin{equation}\label{eq:Hyperelasticity-b}
  \bsigma := - \frac{2}{\gamma} \rho\, (\mg)^{\sharp} \bF^{-\star}\frac{\partial  e}{\partial \bH} \bF^{-1} (\mg)^{\sharp},
\end{equation}
a relation which can be interpreted as a three-dimensional Hyperelasticity law. Introducing~\eqref{eq:def-s-3D}, we end up with
\begin{equation}\label{eq:Hyperelasticity-b-bis}
  \bsigma :=  \frac{1}{\gamma} \rho\, (\mg)^{\sharp} \bF^{-\star}\,\bs\, \bF^{-1} (\mg)^{\sharp},
  \qquad
  \bs=\displaystyle-2\frac{\partial e}{\partial \bH}.
\end{equation}

\bigskip

Conversely, once the three-dimensional generalized Cauchy stress tensor $\bsigma$ is given (through a three-dimensional Hyperelasticity law), the four-dimensional stress tensors $\bS$ and $\bSigma$ are then fully determined, either by~\eqref{eq:S-from-sigma} or by~\eqref{eq:Sigma-from-sigma}. Even if the full theory is four-dimensional, \emph{the Relativistic Hyperelasticity laws are three-dimensional}. Finally, observe also that the difference between~\eqref{eq:Hyperelasticity-a} and~\eqref{eq:Hyperelasticity-b} is purely relativistic: both of them converge to the same three-dimensional stress tensor $\bsigma$ at the Galilean limit $c\to 0$ if one assumes $\lim_{c\to \infty}(E/c^{2})=0$.

\begin{rem}\label{rem:Kirchhoff}
  The three-dimensional stress tensor,
  \begin{equation*}
    \btau:=\frac{\bsigma}{\rho_{r}}=\gamma\frac{\bsigma}{\rho},
    \quad\text{such as} \quad
    \btau =(\mg)^{\sharp} \bF^{-\star}\,\bs\, \bF^{-1} (\mg)^{\sharp},
  \end{equation*}
  is a first Relativistic generalization of the Kirchhoff stress tensor $\bsigma/\rho$ of Classical Continuum Mechanics (see \autoref{sec:stress-on-the-body} for a second generalization relative to the Schwarzschild spacetime).
\end{rem}

% ----------------------------------------------------------------
\section{Relativistic Hyperelasticity in Schwarzschild spacetime}
\label{sec:hyperelasticity-Schwarzschild}
% ----------------------------------------------------------------

In~\cite{Sou1958,Sou1960,Sou1964}, Souriau has provided a full description of Relativistic Hyperelasticity in the Minkowski spacetime, making implicitly the \emph{passive matter hypothesis}, meaning that the matter field $\Psi$ under study is negligible as a source of the gravitation field. Minkowski spacetime is a flat static spacetime, with no gravitational source, which is described in some coordinate system $(x^{0}=ct, x^{i})$ by the constant metric
\begin{equation*}
  \eta = -(dx^{0})^{2} + \delta_{ij}dx^{i}dx^{j}.
\end{equation*}
Of course, this situation is not fully realistic. However, due to the fact that for any point $m_{0}$ of the Universe, it is always possible to find a chart around $m_{0}$ in which the Christoffel symbols vanish at $m_{0}$, we can assume that the Christoffel symbols are almost zero around this point. This situation exactly corresponds to a free fall (like inside an orbital station), it approximately corresponds to mechanical situations on Earth surface for which gravity can be neglected or not taken into account.

Our goal here is to extend Souriau's results on Relativistic Hyperelasticity by taking into account gravity. These results will be used in the next section to formulate classical Galilean Hyperelasticity with gravity (or Newton--Cartan theory of Continuum Mechanics~\cite{Car1923,Car1924,Car1925}). To do so, rather than using the Minkowski spacetime, we shall assume here that the continuous medium/the structure considered is embedded in the \emph{Schwarzschild spacetime}~\cite{Sch1916,MTW1973,MG2009} (see \autoref{sec:Schwarzschild-spacetime}). Therefore, we neglect the influence of the matter under study (passive matter hypothesis) as a source of the gravity field. The \emph{exterior Schwarzschild metric} is a static solution $g$ of Einstein equation $\Ein_{g} = 0$ in the vacuum (with vanishing cosmological constant $\Lambda=0$). It is representative of the gravity field around a spherical and nonrotating planet (or a star or a black hole) of mass $M$ and radius $r_{0}$, such as the Earth. In this model, the rotation of the celestial body as a potential source of the gravitation field has been neglected. An alternative would have been to choose the \emph{Kerr metric}~\cite{MG2009} rather than the Schwarzschild metric, another possible choice which we did not make. A practical consequence of this choice is that the frame deduced from the coordinate system in which is described the
Schwarzschild metric corresponds to one pointing at fixed stars (the Earth is assumed to be nonrotating).

In the so-called \emph{Cartesian isotropic coordinates} (see \autoref{sec:Schwarzschild-spacetime}), the Schwarzschild metric has for expression
\begin{equation*}
  g = -\left(\frac{\bar{r}-\bar{r}_{\mathrm{s}}}{\bar{r} + \bar{r}_{\mathrm{s}}}\right)^{2} c^{2} \dd t^{2}
  + \left(1+\frac{\bar{r}_{\mathrm{s}}}{\bar{r}}\right)^{4} \delta_{ij} x^{i}x^{j},
  \qquad
  \bar{r}:=\sqrt{\delta_{ij} x^{i}x^{j}},
\end{equation*}
with $\bar{r}=0$ at the center of the planet and $\bar{r}\approx r_{0}$ on its surface. The \emph{reduced Schwarzschild radius}
\begin{equation*}
  \bar{r}_{\mathrm{s}} = \frac{1}{4} r_{\mathrm{s}} = \frac{GM}{2c^{2}}
\end{equation*}
depends on the gravitational constant $G$, the mass $M$ of the celestial body, and the light speed $c$. It is much smaller than the radius $r_{0}$ of the planet~\cite{MTW1973} ($\bar{r}_{\mathrm{s}}\approx 2$ mm for the Earth).

Introducing the \emph{lapse function}~\cite{ADM1962}
\begin{equation*}
  \mathcal{N}:=\frac{1}{\sqrt{-\norm{\grad^{g} x^{0}}^{2}}} = \frac{1}{\sqrt{-g^{00}}} = \sqrt{-g_{00}} =  \frac{1-\bar{r}_{\mathrm{s}}/\bar{r}}{1+\bar{r}_{\mathrm{s}}/\bar{r}},
\end{equation*}
the metric can be written as
\begin{equation}\label{eq:Schwarzschild-metric}
  g = - \mathcal{N}^{2} c^{2}  \dd t^{2}+ \mg,
\end{equation}
where $\mg$ is the spatial (conformal) metric,
\begin{equation}\label{eq:k}
  \mg = k q,
  \qquad
  k = \left(1+\frac{\bar{r}_{\mathrm{s}}}{\bar{r}}\right)^{4}
  \qquad
  q = \delta_{ij} \dd x^{i} \dd x^{j}.
\end{equation}
The unit normal $\bN$, defined by~\eqref{eq:defN}, to the three-dimensional spatial hypersurfaces $\Omega_{t}$ is simply
\begin{equation}\label{eq:Schwarzschild-normal}
  \bN = \frac{1}{\mathcal{N}}\, \partial_{x^{0}}=\frac{1}{c\mathcal{N}}\, \partial_{t}.
\end{equation}

\begin{rem}
  The flat Minkowski spacetime (Special Relativity) is simply the limiting massless case $M=0$ and thus $\bar{r}_{\mathrm{s}} = r_{\mathrm{s}} = 0$. It is a special case of this more general framework.
\end{rem}

The Schwarzschild metric is not flat. The non-vanishing symmetric Christoffel symbols, in the Cartesian isotropic coordinate systems $(x^{\mu})$ and $(t, x^{i})$, can easily be recovered using the usual formula~\eqref{eq:Gammas} since $\mg=k(\bar{r}) q$ is conformal (see also~\cite{MG2009}). They are written as
\begin{equation}\label{eq:Schwarzschild-Christoffels}
  \begin{split}
    \Gamma^{t}_{ti} & = \Gamma^{0}_{0i} = (\dd \ln \mathcal{N})_{i} = 2 \frac{\bar{r}_{\mathrm{s}}}{\bar{r}^{2}} \frac{ 1}{\left(1-\bar{r}_{\mathrm{s}}^{2}/\bar{r}^{2}\right)} \delta_{ik} \frac{x^{k}}{\bar{r}},
    \\
    \Gamma^{i}_{tt} & = c^{2} \Gamma^{i}_{00} = c^{2}\frac{\mathcal{N}}{k} (\grad \mathcal{N})^{i} = 2c^{2} \frac{\bar{r}_{\mathrm{s}}}{\bar{r}^{2}} \frac{(1-\bar{r}_{\mathrm{s}}/\bar{r})}{(1+\bar{r}_{\mathrm{s}}/\bar{r})^{7}} \frac{x^{i}}{\bar{r}},
    \\
    \Gamma^{i}_{jj} & = -\frac{1}{2}(\grad \ln k)^{i} = 2 \frac{\bar{r}_{\mathrm{s}}}{\bar{r}^{2}} \frac{ 1}{\left(1+\bar{r}_{\mathrm{s}}/\bar{r}\right)} \frac{x^{i}}{\bar{r}}, \quad \text{for $i\ne j$},
    \\
    \Gamma^{j}_{ji} & = \frac{1}{2}(\dd \ln k)_{i} = - 2 \frac{\bar{r}_{\mathrm{s}}}{\bar{r}^{2}} \frac{ 1}{\left(1+\bar{r}_{\mathrm{s}}/\bar{r}\right)}\delta_{ik}\frac{x^{k}}{\bar{r}}, \quad \text{for $i = j$ and $i\ne j$}.
  \end{split}
\end{equation}
with no sum on $j$, and where, for a function $f$, $\grad f=q^{-1} \dd f$ is the gradient relative to the Euclidean metric $q=(\delta_{ij})$. The related divergence operators are detailed in~\autoref{sec:divergences-spacetime}.

In the following, we particularize the relations established in \autoref{sec:hyperelasticity-spacetime} for the special case of the Schwarzschild spacetime when expressed in Cartesian isotropic coordinates.

$\bullet$ The \emph{three-dimensional velocity} $\uu$ defined by~\eqref{eq:def-u} is then given by
\begin{equation}\label{eq:u-s}
  \uu = -\frac{1}{\mathcal{N}}\, \bF\, \frac{\partial \Psi}{\partial t},
\end{equation}
where $\bF=T\Psi_{t}^{-1}$ was introduced in~\eqref{eq:def-u}.

$\bullet$ The \emph{generalized Lorentz factor}~\eqref{eq:def-gamma} has then for expression
\begin{equation}\label{eq:gamma-s}
  \gamma = \frac{1}{\sqrt{1-k \frac{\uu^{2}}{c^{2}}}},
\end{equation}
where $\uu^{2}:=\norm{\uu}_{q}^{2}$ is the Euclidean squared norm and $k=\left(1+\frac{\bar{r}_{\mathrm{s}}}{\bar{r}}\right)^{4}$.

$\bullet$ The \emph{conformation}~\eqref{eq:conformation-in-spacetime} reduces to
\begin{equation}\label{eq:conformation-s}
  \bH=\bF^{-1}\left( \frac{1}{k}\bq^{-1}-\frac{1}{c^{2}} \uu \otimes \uu\right) \bF^{-\star}.
\end{equation}

$\bullet$ For the two choices of a \emph{three-dimensional stress} $\bsigma$ introduced in \autoref{sec:hyperelasticity-spacetime}, where $\rho=\gamma\rho_{r}$ is the relativistic mass density, we get, in the coordinate system $(t, x^{i})$,
\begin{equation}\label{eq:Ttxi}
  \bT =\left(\begin{array}{cc}\displaystyle
      \frac{1}{c^{2}\mathcal{N}^{2}} E_{\text{tot}} & \displaystyle\frac{1}{c^{2} \mathcal{N}} \bp^{\star}
      \\
      \displaystyle\frac{1}{c^{2} \mathcal{N}} \bp  & \bs
    \end{array}\right)
\end{equation}

\begin{enumerate}
  \item When the stress tensor $\bsigma$ is defined as the spatial part of $\bS$:
        \begin{equation}\label{eq:T1-s}
          \begin{cases}
            E_{\mathrm{tot}}=  \gamma \rho c^{2} + E\Big(1+k \frac{\uu^{2}}{c^{2}} \Big) - \frac{1}{c^{2}} \uu^{\flat} \cdot \bsigma \cdot \uu^{\flat},
            \\
            \bp = (\gamma \rho c^{2}+E)\uu- \bsigma \cdot \uu^{\flat},
            \\
            \bt = \gamma\rho \, \uu \otimes \uu- \bsigma,
            \\
            \bsigma = - \frac{2}{\gamma k^{2}} \rho\, q^{-1} \bF^{-\star}\frac{\partial  e}{\partial \bH} \bF^{-1} q^{-1}
            - \frac{\gamma^{2} E}{c^{2}} \uu \otimes \uu.
          \end{cases}
        \end{equation}

  \item When the stress tensor $\bsigma$ is defined as the spatial part of $\bSigma$:
        \begin{equation}\label{eq:T2-s}
          \begin{cases}
            E_{\mathrm{tot}}=  \gamma \rho c^{2} + E\Big(1+k \gamma^{2} \frac{\uu^{2}}{c^{2}}\Big)-\frac{1}{c^{2}} \uu^{\flat} \cdot \bsigma \cdot \uu^{\flat}
            ,
            \\
            \bp = \left(\gamma \rho c^{2} + E\Big(1+k \gamma^{2}\frac{\uu^{2}}{c^{2}} \Big)\right) \uu - \bsigma \cdot\uu^{\flat},
            \\
            \bt = \left(\gamma \rho + \frac{1}{c^{2}}E \Big(1+k\gamma^{2} \frac{\uu^{2}}{c^{2}} \Big)\right) \uu \otimes \uu- \bsigma,
            \\
            \bsigma = - \frac{2}{\gamma k^{2}} \rho\, q^{-1} \bF^{-\star}\frac{\partial  e}{\partial \bH} \bF^{-1} q^{-1}.
          \end{cases}
        \end{equation}
\end{enumerate}

\begin{rem}
  The above expressions generalize Souriau's 1958--1964 results when gravity is taken into account. Note that in~\eqref{eq:T1-s}, in the special case $k=1$ (Special Relativity), there is a contribution $E\uu^{2}/c^{2}$ in the expression of $E_{\mathrm{tot}}$. This term is however missing in~\cite[p. 153]{Sou1958} but this typo was corrected in~\cite[p. 376]{Sou1964}.
\end{rem}

The full description of Relativistic Hyperelasticity must be completed by writing balance laws for the four-momentum quadrivector and the stress-energy tensor. They are written as
\begin{align}
  \label{eq:divP=0}
  \dive^{g} \bP & = 0,
  \\
  \label{eq:divT=0}
  \dive^{g} \bT & = 0,
\end{align}
where $\bP$ is the current of matter and $\bT$ is the stress-energy tensor. Using explicit formulas for these divergence operators provided in \autoref{sec:divergences-spacetime}, where we use here the Cartesian isotropic coordinates system $(t, x^{i})$, we obtain the following equations.

$\bullet$ The conservation law~\eqref{eq:divP=0} for the \emph{current of matter}
\begin{equation}\label{eq:P-s}
  \bP = \rho_{r} \bU = \frac{\rho}{c\mathcal{N}} \partial_{t} + \frac{1}{c}\rho\uu,
\end{equation}
leads (after multiplication by $c$) and according to~\eqref{eq:divP-static}, to
\begin{equation}\label{eq:Schwarzschild-divP}
  \dive^{g} c\bP = \frac{1}{\mathcal{N}} \frac{\partial \rho}{\partial t}+ \frac{\partial}{\partial x^{i}} (\rho u^{i})+  (\Gamma^{j}_{j i } +\Gamma^{t}_{t i } ) \rho u^{i}=0,
\end{equation}
or, in a more intrinsic form
\begin{equation*}
  \dive^{g} c\bP = \frac{1}{\mathcal{N}} \frac{\partial \rho}{\partial t} + \dive^{\mg} (\rho \uu) + \rho \uu\cdot \dd \ln \mathcal{N} = 0.
\end{equation*}

$\bullet$ The conservation law~\eqref{eq:divT=0} for the \emph{stress-energy tensor} $\bT$, expressed in components by~\eqref{eq:Ttxi},
has for time component, according to~\eqref{eq:divT0-static}
\begin{equation}\label{eq:Schwarzschild-divT-t}
  (\dive^{g} \bT)^{t} =  \frac{1}{c^{2}\mathcal{N}^{2}}\frac{\partial E_{\mathrm{tot}}}{\partial t} + \frac{1}{c^{2}\mathcal{N}} \frac{\partial p^{i}}{\partial x^{i}}
  + \left(  \Gamma^{j}_{j i} + 3\Gamma^{t}_{ti} - \frac{\partial \mathcal{\ln N}}{\partial x^{i}}\right) \frac{p^{i}}{c^{2}\mathcal{N}} = 0,
\end{equation}
or, in a more intrinsic form
\begin{equation*}
  (\dive^{g} \bT)^{t} = \frac{1}{c^{2}\mathcal{N}^{2}}\frac{\partial E_{\mathrm{tot}}}{\partial t} + \frac{1}{c^{2}\mathcal{N}} \dive^{\mg} \bp+ 2 \frac{\bp}{c^{2}\mathcal{N}} \cdot \dd \ln \mathcal{N} =0,
\end{equation*}
and for spatial component, according to~\eqref{eq:divTi-static}
\begin{equation}  \label{eq:Schwarzschild-divT-i}
  (\dive^{g} \bT)^{i} =  \frac{1}{c^{2}\mathcal{N}} \frac{\partial p^{i}}{\partial t} + \frac{\partial s^{i j}}{\partial x^{j}}
  + (\Gamma^{k}_{k j}+\Gamma^{0}_{0 j}) s^{ij}+ \Gamma^{i}_{jk}  s^{jk}+ \Gamma^{i}_{tt} \frac{E_{\mathrm{tot}}}{c^{2}\mathcal{N}^{2}} =0,
\end{equation}
or, in a more intrinsic form
\begin{equation*}
  (\dive^{g} \bT)^{\top} =  \frac{1}{c^{2}\mathcal{N}} \frac{\partial \bp}{\partial t} + \dive^{\mg} \bs
  +   \bs \cdot \dd \ln \mathcal{N}
  + E_{\mathrm{tot}} \grad^{\mg}\!\! \ln \mathcal{N} =0,
\end{equation*}
where $ \grad^{\mg}\!\! \ln \mathcal{N}= \frac{1}{k} q^{-1} \dd \ln \mathcal{N}$.

% ----------------------------------------------------------------
\section{The Galilean limit of Relativistic Hyperelasticity}
\label{sec:Gallilean}
% ----------------------------------------------------------------

Soon after Einstein's formulation of the theory of General Relativity (1915), Cartan introduced, in a series of papers~\cite{Car1923,Car1924,Car1925}, a general covariant formulation of Newtonian gravity, today called \emph{Newton-Cartan theory of gravitation}. It can be obtained as the limit of Lorentzian spacetimes whose light cones open up to hyperplanes at each tangent space~\cite{Kuen1976}. It allows to recast the equations of energy and momentum balance of Classical Continuum Mechanics in a four-dimensional general covariant form, similar to the relativistic equation $\dive^{g} \bT = 0$, provided the specific internal energy and the energy flux are suitably interpreted. This 4D general covariant formalism~\cite{Tou1958,Kuen1972,Dix1975,DK1977,Hav1964,DBKP1985,dSax2020}, derived from General Relativity, is also useful to better understand the foundations of Classical Continuum Mechanics and avoid \emph{ad hoc} assumptions in its formulation.

A \emph{Galilean structure} on a four-dimensional manifold $\mM$ is a pair $(\got,\theta)$, where $\got$ is a symmetric second-order contravariant tensor of signature $(0,+,+,+)$ (the classical spatial metric) and $\theta$ is a one-form which spans the kernel of $\got$ (the clock). This means that $\got \theta = 0$ and that $\theta$ vanishes nowhere. The one-form $\theta$ defines a distribution of hyperplanes $E_{m} := \ker \theta_{m}$, each of them, carrying an Euclidean metric induced by $\got$. The Galilean structure is called integrable if $\theta$ is closed. In that case, $\theta$ defines a time function $\hat{t}$ satisfying $\dd\hat{t} = \theta$ (at least locally) and a foliation by hypersurfaces, $\espace_{t} := \hat{t}^{-1}(t)$, tangent to the distribution $(E_{m})$, and which are moreover Riemannian manifolds.

A covariant derivative $\nabla$ on $\mM$ is said to be Galilean if it is symmetric (torsion-less) and satisfies moreover
\begin{equation*}
  \nabla \got = 0 \quad \text{and} \quad \nabla \theta = 0.
\end{equation*}
Such a covariant derivative exists only if the Galilean structure $(\got,\theta)$ is integrable (since $\nabla \theta = 0$ implies $\dd\theta = 0$ for a symmetric covariant derivative). Note, however, that contrary to the canonical covariant derivative of a Riemannian (or pseudo-Riemannian) manifold, a Galilean covariant derivative is not uniquely defined.

In practice, a Galilean structure $(\got,\theta)$ is obtained as a limit of a one-parameter family of smooth Lorentz metrics $\gl$, such that
\begin{equation*}
  \gl^{-1} = \got + \lambda \bkappa + O(\lambda^{2}),
\end{equation*}
with $\got$ of signature $(0,+,+,+)$, and $\theta$ is a generator of the kernel of $\got$~\cite{Dau1964,Kuen1972}. Note that $\theta$ can be fixed uniquely up to a sign by the normalization $\bkappa(\theta,\theta) = -1$. It has been shown in~\cite{Kuen1976} that, provided that $\dd\theta = 0$ (a condition which is always satisfied by static spacetimes such as the Minkowski or the Schwarzschild spacetimes), the one-parameter family of Riemannian covariant derivatives $\nabla^{\lambda}$ converges then to a symmetric covariant derivative $\nabla^{\text{\tiny NC}}$ which is compatible with the Galilean structure $(\got,\theta)$.

Recall moreover, that, on a Riemannian or pseudo-Riemannian manifold, the Riemann tensor $\mathbf{R}$ (defined in components by $R_{\alpha\beta\gamma\delta} = g_{\alpha\rho}{R^{\rho}}_{\beta\gamma\delta}$) has the additional symmetry
\begin{equation*}
  R_{\alpha\beta\gamma\delta} = R_{\gamma\delta\alpha\beta},
\end{equation*}
and that, uprising the first and third indices, we obtain the following identity
\begin{equation*}
  g^{\gamma\lambda}{R^{\alpha}}_{\beta\lambda\delta} = g^{\alpha\lambda}{R^{\gamma}}_{\delta\lambda\beta}.
\end{equation*}
Therefore, when a Galilean covariant derivative is obtained as a limit of (pseudo-)Riemannian covariant derivatives, it must satisfies the additional property
\begin{equation}\label{eq:Newtonnian-connexion}
  \got^{\gamma\lambda}{R^{\alpha}}_{\beta\lambda\delta} = \got^{\alpha\lambda}{R^{\gamma}}_{\delta\lambda\beta}
\end{equation}
and is then called a \emph{Newtonian covariant derivative}.

Applying this procedure to the Schwarzschild metric~\eqref{eq:Schwarzschild-metric},
\begin{equation*}
  \gl = - \Nl^{2} c^{2}  \dd t^{2}+ \kl \, q,
\end{equation*}
with $\lambda := 1/c^{2}$, we get
\begin{equation}\label{eq:Nc4}
  \Nl=1-\frac{1}{c^{2}} \frac{GM}{\bar{r}}+O(1/c^{4}),
  \qquad
  \kl=1+\frac{2}{c^{2}} \frac{GM}{\bar{r}}+O(1/c^{4}),
\end{equation}
and
\begin{equation}\label{eq:metric-expansion}
  \gl^{-1} = q^{-1} - \frac{1}{c^{2}} \left( (\partial_{t})^{2} + 2\frac{GM}{\bar{r}}q^{-1}\right) + O(1/c^{4}),
\end{equation}
where $q^{-1} = (\partial_{x^{1}})^{2} + (\partial_{x^{2}})^{2} + (\partial_{x^{3}})^{2}$. Observe also that we have,
\begin{equation}\label{eq:volg-expansion}
  \vol_{\gl} = c f\, \dd t \wedge \vol_{q},
  \quad \text{where} \quad  f=\Nl \kl^\frac{3}{2} =1+\frac{2}{c^{2}} \frac{GM}{\bar{r}}  + O(1/c^{4}),
\end{equation}
for the Riemannian volume form associated with the metric $\gl$. Note that $\vol_{\gl}$ diverges as $c \to \infty$.

We obtain thus the following Galilean structure,
\begin{equation*}
  \got = q^{-1}, \qquad \bkappa = - (\partial_{t})^{2} - 2\frac{GM}{\bar{r}}q^{-1}, \qquad \theta = \dd t,
\end{equation*}
where the normalisation condition $\bkappa(\theta,\theta) = -1$ has been used. This structure is of course integrable, and the time function is the same, in either the relativistic context or the Galilean one. Therefore, the foliation by the hypersurfaces $\Omega_{t}$ is common to both structures. Note however that by~\eqref{eq:Schwarzschild-normal}, the relativistic normal $\bN^{\lambda}$ to these hypersurfaces converges towards $0$ as $c \to \infty$.

An immediate consequence of~\eqref{eq:metric-expansion} is the fact that the \emph{conformation} $\Hl$ defined by~\eqref{eq:conformation} has a limit when $c\to \infty$. Indeed,
\begin{equation}\label{eq:HapproxCm1}
  \Hz := \lim_{c\to \infty}\Hl = \bF^{-1}\bq^{-1} \,\bF^{-\star}.
\end{equation}

\begin{rem}
  Observe the similarity between the (limit) conformation $\Hz$ and the inverse $\bC^{-1}$ of right Cauchy--Green tensor $\bC:=\bF^{\star} \bq\, \bF$ in Classical Continuum Mechanics. However, $\Hz$ is not exactly $\bC^{-1}$ because it is a function from the World tube $\mW$ to $\Sym^{2}T\body$ (see remark~\ref{rem:Psim1}), while $\bC^{-1}$ is a tensor field on $\body$.
\end{rem}

Concerning the Riemannian covariant derivative $\nabla^{\lambda}$ of $\gl$, the expansion of its non vanishing Christoffel symbols is easily deduced from~\eqref{eq:Schwarzschild-Christoffels} and recalling that $\bar{r}_{\mathrm{s}} = GM/2c^{2}$. We get
\begin{equation}\label{eq:Schwarzschild-Christoffels-expansions}
  \begin{split}
    \Gamma^{i}_{tt} & =  2c^{2} \frac{\bar{r}_{\mathrm{s}}}{\bar{r}^{2}} \frac{(1-\bar{r}_{\mathrm{s}}/\bar{r})}{(1+\bar{r}_{\mathrm{s}}/\bar{r})^{7}} \frac{x^{i}}{\bar{r}} = -\mathrm{g}^{i} -\frac{1}{c^{2}}\frac{4GM}{\bar{r}^{2}}\mathrm{g}^{i} + O(1/c^{4}),
    \\
    \Gamma^{t}_{ti} & =  2 \frac{\bar{r}_{\mathrm{s}}}{\bar{r}^{2}} \frac{ 1}{\left(1-\bar{r}_{\mathrm{s}}^{2}/\bar{r}^{2}\right)} \delta_{ik} \frac{x^{k}}{\bar{r}} = -\frac{1}{c^{2}}\delta_{ik}\mathrm{g}^{k}+O(1/c^{4}),
    \\
    \Gamma^{i}_{jj} & = 2 \frac{\bar{r}_{\mathrm{s}}}{\bar{r}^{2}} \frac{ 1}{\left(1+\bar{r}_{\mathrm{s}}/\bar{r}\right)} \frac{x^{i}}{\bar{r}} = - \frac{1}{c^{2}}\mathrm{g}^{i}+O(1/c^{4}), \quad \text{for $i\neq j$},
    \\
    \Gamma^{j}_{ji} & = - 2 \frac{\bar{r}_{\mathrm{s}}}{\bar{r}^{2}} \frac{ 1}{\left(1+\bar{r}_{\mathrm{s}}/\bar{r}\right)}\delta_{ik}\frac{x^{k}}{\bar{r}} = \frac{1}{c^{2}}\delta_{ik}\mathrm{g}^{k}+O(1/c^{4}), \quad \text{for $i = j$ and $i\ne j$}.
  \end{split}
\end{equation}
with no sum on $j$, and where $\mathbf{g}$ is the Newtonian (centripetal) gravity field,
\begin{equation}\label{eq:gxi}
  \mathbf{g} := -\frac{GM} {\bar{r}^{2}} \frac{x^{i}}{\bar{r}} \partial_{x^{i}},
\end{equation}
with $\frac{GM} {\bar{r}^{2}}\approx \frac{GM} {r_{0}^{2}} = 9.8\, \text{m/s}^{2}$ on Earth surface. We deduce therefore that the Christoffel symbols of the Newton--Cartan limit $\nabla^{\text{\tiny NC}}$ are all vanishing except
\begin{equation}\label{eq:lmitGammaSchw}
  \Gamma^{i}_{tt} = -\mathrm{g}^{i}.
\end{equation}

\begin{rem}[Weak gravity]\label{rem:gravity}
  The approximations~\eqref{eq:Schwarzschild-Christoffels-expansions} at large $c$ coincide in fact with those corresponding to so-called \emph{weak gravity} (see \cite{Der2006}), valid for the Earth, and more generally for a planet or a star (any object whose radius is smaller than its Schwarzschild radius being called a black hole).
\end{rem}

The divergence of a quadrivector $\bP$, relative to the Newtonian covariant derivative $\nabla^{\text{\tiny NC}}$, is given by
\begin{equation}\label{eq:P-Galilean-divergence}
  \divz \bP = \frac{\partial P^{t}}{\partial t} + \frac{\partial P^{j}}{\partial x^{j}}.
\end{equation}
It corresponds to the zero order terms in the expansions in $\lambda=1/c^{2}$ of the divergence $\divl \bP$ relative to the Schwarzschild metric $\gl$ given by~\eqref{eq:divP-static}, since we have
\begin{equation*}
  \Gamma^{j}_{j i } +\Gamma^{t}_{t i }=\frac{2}{c^{2}} \delta_{ik} \mathrm{g}^{k}+ O(1/c^{4}).
\end{equation*}
We get therefore
\begin{equation}\label{eq:divPc4}
  \divl \bP=\divz \bP+\frac{2}{c^{2}} \mathbf{g}\cdot \bP^{\top}+O(1/c^{4}),
\end{equation}
where $\bP^{\top}$ is the spatial part of the quadrivector $\bP$.

The divergence of a symmetric second-order contravariant tensor $\bT$ is given by
\begin{equation}\label{eq:T-Galilean-divergence}
  (\divz \bT)^{t} =  \frac{\partial T^{tt}}{\partial t} + \frac{\partial T^{tj}}{\partial x^{j}},
  \qquad (\divz \bT)^{i} =  \frac{\partial T^{it}}{\partial t} +  \frac{\partial T^{ij}}{\partial x^{j}} - \mathrm{g}^{i}T^{tt}.
\end{equation}
Indeed,  by~\eqref{eq:divT0-static}, and since by~\eqref{eq:Schwarzschild-Christoffels},
\begin{equation}\label{eq:sumGammai}
  \sum_{j=1}^{3}\Gamma^{j}_{ji}+3 \Gamma^{t}_{ti} =3\left(\dd \ln(\mathcal{N}\sqrt{k})\right)_{i}= O(1/c^{4}),
\end{equation}
we obtain
\begin{equation}\label{eq:divTtc4}
  (\divl \bT)^{t} = (\divz \bT)^{t}+O(1/c^{4}),
\end{equation}
whereas~\eqref{eq:divTi-static-components}, combined with the fact that all involved Christoffel's symbols are $O(1/c^{2})$, but
\begin{equation*}
  \Gamma^{i}_{tt} = -\mathrm{g}^{i} + O(1/c^{2}),
\end{equation*}
ends up to
\begin{equation*}
  (\divl \bT)^{i} = \frac{\partial T^{it}}{\partial t} + \frac{\partial T^{ij}}{\partial x^{j}} - \mathrm{g}^{i}T^{tt}+ O(1/c^{2}) = (\divz \bT)^{i} + O(1/c^{2}).
\end{equation*}

The \emph{current of matter}~\eqref{eq:def-P} $\Pl$ for the metric $\gl$ is defined implicitly by
\begin{equation*}
  i_{\bP^{\lambda}} \vol_{\gl}=\omega=\Psi^{*} \mu,
\end{equation*}
where the 3-form $\omega$ does not depend on the light speed $c$ (the matter field $\Psi$ and the mass measure $\mu$ do not depend on $c$, which is only introduced through the metrics). By~\eqref{eq:volg-expansion}, we get thus
\begin{equation*}
  cf\,  i_{\bP^{\lambda}} \,\dd t\wedge \vol_{q}= i_{cf \bP^{\lambda}} \,\dd t\wedge \vol_{q}=\omega=\Psi^{*} \mu,
\end{equation*}
from which it is seen that $cf \Pl$ is independent of $c$, and is therefore equal to its Newtonian limit $(c\bP)^{\scriptscriptstyle 0}$ defined by $i_{(c\bP)^{\scriptscriptstyle 0}} \,\dd t\wedge \vol_{q}=\omega$ (since $f \to 1$ as $c\to \infty$).
Setting
\begin{equation}\label{eq:cPLcP0}
  c\Pl = \rhol \,\partial_{t} + \rhol\, \uul,
  \qquad
  (c\bP)^{\scriptscriptstyle 0} = \rhoz\, \partial_{t} + \rhoz\, \uuz,
\end{equation}
defines the mass density $\rhol$ and the velocity $\uul$, as well as their Newtonian limits $\rhoz$ and $\uuz$. The equality $cf \bP^{\lambda}=(c\bP)^{\scriptscriptstyle 0}$ leads to
\begin{equation*}
  \rhol = \frac{1}{f} \rhoz,\qquad \uul = \uuz.
\end{equation*}
and allows us to expand the mass density as
\begin{equation*}
  \rhol = \rhoz- \frac{2}{c^{2}} \frac{GM}{\bar{r}} \rhoz + O(1/c^{4}).
\end{equation*}
We have therefore
\begin{equation*}
  \lim_{c\to\infty} \rhol = \rhoz,
  \qquad
  \lim_{c\to\infty} \uul = \uuz,
\end{equation*}
and
\begin{equation*}
  \lim_{c\to\infty} \gamma=\lim_{c\to\infty} \frac{1}{\sqrt{1-k \frac{\uu^2}{c^{2}}}}=1,
\end{equation*}
since the \emph{generalized Lorentz factor}~\eqref{eq:gamma-s} has the classical expansion
\begin{equation*}
  \gamma = 1+\frac{1}{2} \frac{\uuz^2}{c^{2}} + O(1/c^{4}).
\end{equation*}

Using~\eqref{eq:P-s}, we see that $(\rhol,\uul)$ are connected to $(\rho,\uu)$ by
\begin{equation*}
  \rhol=\frac{\rho}{\Nl},
  \qquad
  \uul=\Nl \uu ,
  \qquad
  \rhol\,\uul=\rho \uu,
\end{equation*}
when the so-called relativistic mass density $\rho$ and velocity $\uu$ are defined by the orthogonal decomposition
\begin{equation*}
  c\Pl=\rho c\, {\overset{\scriptscriptstyle\lambda}{\bN}}+\rho \uu.
\end{equation*}
By~\eqref{eq:u-s}, we deduce that the Newtonian limit of the three-dimensional velocity $\uu=\uul/\Nl$ is
\begin{equation*}
  \lim_{c\to \infty} \uu =\uuz =  - \bF \frac{\partial \Psi}{\partial t}.
\end{equation*}

\begin{rem}
  In Classical Continuum Mechanics, the \emph{Eulerian velocity} is defined as the vector field on the deformed configuration given by
  $\partial_{t} \pp \circ \pp^{-1}$, where $\pp$ is the embedding of the body $\body$ into the Euclidean space, and where $\VV:=\partial_{t} \pp$ is the \emph{Lagrangian velocity}. If we assume, furthermore, that $\Psi_{t}$ is a diffeomorphism and we set $\pp=\Psi_{t}^{-1}$ (see remark \autoref{rem:TPsit-iso}), the vector field $\uuz$ recasts as
  \begin{equation*}
    \uuz = - \bF\,\partial_{t} \Psi = \partial_{t} \Psi_{t}^{-1} \circ \Psi_{t} =
    \partial_{t} \pp \circ \pp^{-1},
  \end{equation*}
  and we recognize $\uuz$ as the Eulerian velocity of Classical Continuum Mechanics.
\end{rem}

The \emph{stress-energy tensor} $\Tl:=\bT$ has for expression, in the coordinate system $(t, x^{i})$,
\begin{equation}\label{eq:TtxiN}
  \Tl =\begin{pmatrix}
    \displaystyle \frac{1}{c^{2} \Nl} \El_{\text{tot}} & \displaystyle \frac{1}{c^{2} \Nl} \pl^{\star}
    \\
    \displaystyle\frac{1}{c^{2} \Nl} \pl               & \displaystyle \bsl
  \end{pmatrix}
  \qquad\text{where}\qquad
  \begin{cases}
    \El_{\text{tot}}:=E_{\text{tot}}/\mathcal{N},
    \\
    \pl:= \bp,
    \\
    \bsl:=\bs.
  \end{cases}
\end{equation}
To determine its limit, observe that the energy density $\El_{\text{tot}}$, the linear momentum $\pl$ and the spatial part $\bsl$ behave as
\begin{align}\label{eq:Etotlambda}
  \frac{\El_{\text{tot}}}{c^{2}} &
  = \rhol+ \frac{1}{c^{2}}\left(E+\frac{1}{2} \rho \uu^{2}\right) +  O(1/c^{4}),
  \\
  \label{eq:plambda}
  \frac{\pl}{c^{2}}              & =   \rhol \,\uul + \frac{1}{c^{2}} \left(\left(E+\frac{1}{2} \rho \uu^{2}\right) \uu-\bsigma\cdot\uu^{\flat}\right) + O(1/c^{4}),
  \\
  \label{eq:slambda}
  \bsl                           & = \rhol\,  \uul   \otimes \uul   -\bsigma
  + O(1/c^{2}),
\end{align}
where we have used indifferently either~\eqref{eq:T1-s} or~\eqref{eq:T2-s}, and we have assumed that the internal energy density $E$ (function of $\bH$) is $o(1/c^{2})$ (see also the discussion in \cite{Sou1960}). In that case, the quantities $\frac{1}{2}\rho \uu^{2}$, $E$ and $\bsigma$ converge respectively to $\frac{1}{2}\rhoz\, \uuz^{2}$, $\Ez$ and $\bsigz$.
Therefore, in the coordinate system $(t, x^{i})$, the stress-energy tensor $\Tl$ converges  to the Newtonian limit
\begin{equation}\label{eq:TGal}
  \Tz = \lim_{c\to \infty}\bT^{\lambda}=
  \begin{pmatrix}
    \rhoz        & \rhoz\, \uuz^{\star}              \\
    \rhoz\, \uuz & \rhoz\, \uuz \otimes \uuz- \bsigz
  \end{pmatrix}.
\end{equation}

We now discuss the limits of the balance laws. First, by~\eqref{eq:divPc4} and~\eqref{eq:cPLcP0}, we get
\begin{equation*}
  \divl (c\Pl) = \frac{\partial \rhol}{\partial t}+\dive(\rhol\, \uul)+\frac{2}{c^{2}} \mathbf{g}\cdot \rhol\, \uul+O(1/c^{4})=0,
\end{equation*}
where $\dive$ is the canonical divergence in $\RR^{3}$, and which can be recast as
\begin{equation}\label{eq:divrholul}
  \frac{\partial \rhol}{\partial t}+\dive(\rhol\, \uul)=-\frac{2}{c^{2}} \mathbf{g}\cdot \rho \uul+O(1/c^{4}).
\end{equation}
It converges to
\begin{equation*}
  \divz (c\bP)^{\scriptscriptstyle 0} = \frac{\partial \rhoz}{\partial t} +\dive(\rhoz\, \uuz)
  = 0.
\end{equation*}
One recovers thus the usual expression of mass conservation in Classical Continuum Mechanics (omitting the bars),
\begin{equation}\label{eq:massconservationMMC}
  \frac{\partial \rho}{\partial t} + \dive(\rho \uu)=0,
\end{equation}
for the Euclidean metric $q$.

Then, the relativistic conservation law $\divl \bT^{\lambda} = 0$ converges to the equation $\divz \Tz=0$, with (by~\eqref{eq:T-Galilean-divergence})
\begin{equation*}
  (\divz \Tz)^{t} = \frac{\partial \rhoz}{\partial t} + \frac{\partial}{\partial x^{j}}( \rhoz\, \uz^{j}),
  \quad \text{and} \quad
  (\divz \Tz)^{i} =  \frac{\partial}{\partial t}(\rhoz\, \uz^{i}) +  \frac{\partial}{\partial x^{j}}( \rhoz\, \uz^{i} \uz^{j} - \sigz^{ij}) - \rhoz\, \mathrm{g}^{i},
\end{equation*}
so that
\begin{equation*}
  \frac{\partial \rhoz}{\partial t} +\dive(\rhoz\, \uuz)=0,
  \qquad
  \frac{\partial }{\partial t} (\rhoz\,\uuz)+\dive\left(\rhoz\, \uuz \otimes \uuz-\bsigz\right)-\rhoz \mathbf{g}=0.
\end{equation*}

The first equation is (again) recognized as the mass conservation~\eqref{eq:massconservationMMC} and the second one as the linear momentum balance of Classical Continuum Mechanics, with gravity $\mathbf{g}$ (omitting the bars),
\begin{equation*}
  \frac{\partial }{\partial t} (\rho\,\uu)+\dive\left(\rho\, \uu \otimes \uu-\bsigma\right)-\rho \mathbf{g}=0.
\end{equation*}
By using the mass conservation law, the later can be recast as the classical expression,
\begin{equation}\label{eq:LinearMomentumCCM}
  \rho \left( \frac{\partial \uu}{\partial t}+ \nabla_{\uu} \uu \right)= \dive  \bsigma + \rho \mathbf{g},
\end{equation}
where $\nabla$ is the covariant derivative for the Euclidean metric $q$.

It is furthermore possible to recover the so-called \emph{local form of energy balance} of Classical Continuum Mechanics, as a term of order $O(1/c^{2})$ in the expansion of a combination of both $\divl (c\Pl)=0$ and $(\divl \bT^{\lambda})^{t} = 0$ (see for instance~\cite{Sou1958} for the case of the flat Minkowski spacetime or~\cite{Gou2012} for relativistic fluids in the case of weak gravity or \cite{Kuen1976} for general discussions about this balance law).

The balance law~\eqref{eq:Schwarzschild-divT-t} expresses the vanishing of the time component $(\dive^{g} \bT)^{t}=(\divl \Tl)^{t}=0$  in which here $\bT=\Tl$ is given by~\eqref{eq:TtxiN}.
Since $\El_{\mathrm{tot}}=E_{\mathrm{tot}}/\Nl$ and $\pl=\bp$, it recasts as
\begin{equation*}
  (\dive^{g} \bT)^{t} = (\divl \Tl)^{t} = \frac{1}{\Nl}\left[\frac{\partial }{\partial t} \Big( \frac{\El_{\mathrm{tot}}}{c^{2}}\Big) + \dive \Big( \frac{\pl}{c^{2}}\Big) + \frac{1}{c^{2}}\pl\cdot  \left(3\,\dd \ln(\mathcal{N}\sqrt{k}) - \dd \,\mathcal{\ln N}\right)\right] = 0,
\end{equation*}
where $\dive$ is the divergence relative to the three-dimensional Euclidean metric $q$. We have introduced the 1-form~\eqref{eq:sumGammai}
\begin{equation*}
  3\,\dd \ln(\mathcal{N}\sqrt{k}) := \Big(\sum_{j=1}^{3} \Gamma^{j}_{j i} +3\Gamma^{t}_{ti}\Big)  \dd x^{i}=O(1/c^{4}),
\end{equation*}
which is of order $O(1/c^{4})$,  $\pl$ which is of order $O(c^{2})$ and
\begin{equation*}
  \dd \,\mathcal{\ln N}=-\frac{1}{c^{2}} q\mathbf{g}+O(1/c^{4}).
\end{equation*}
By~\eqref{eq:Nc4}, we get
\begin{equation}\label{eq:divTlc4}
  \Nl(\divl \Tl)^{t} = \frac{\partial }{\partial t} \Big( \frac{\El_{\mathrm{tot}}}{c^{2}}\Big) + \dive \Big( \frac{\pl}{c^{2}}\Big) + \frac{1}{c^{2}}\Big(\frac{\pl}{c^{2}}\cdot \mathbf{g} \Big)+ O(1/c^{4}) = 0.
\end{equation}
Subtracting~\eqref{eq:divrholul} from~\eqref{eq:divTlc4} and using~\eqref{eq:Etotlambda} and~\eqref{eq:plambda}, we get now
\begin{multline*}
  \frac{1}{c^{2}} \left\{ \frac{\partial}{\partial t}\Big[ \Big(E+\frac{1}{2} \rho \uu^{2}\Big) \uu-\bsigma\cdot\uu^{\flat} \Big]
  + \dive \Big[\Big(E+\frac{1}{2} \rho \uu^{2}\Big) \uu-\bsigma\cdot\uu^{\flat}\Big] + \rho\,\mathbf{g}\cdot \uu \right\}
  \\
  - \frac{2}{c^{2}}\rho\,\mathbf{g}\cdot \uu + O(1/c^{4}) = 0,
\end{multline*}
and thus
\begin{equation*}
  \frac{\partial}{\partial t} \Big( E+\frac{1}{2}\rho\uu^{2}\Big)
  +\dive\left(\Big(E+\frac{1}{2}\rho\uu^{2}\Big) \uu -\bsigma\cdot\uu^{\flat}\right)
  -\rho\, \mathbf{g}\cdot \uu=O(1/c^{2}).
\end{equation*}
Passing to the limit $c \to \infty$, we obtain therefore
\begin{equation*}
  \frac{\partial}{\partial t} \Big( \Ez+\frac{1}{2}\rho\uuz^{2}\Big)
  +\dive\left(\Big(\Ez+\frac{1}{2}\rhoz\,\uuz^{2}\Big) \uuz -\bsigz\cdot\uuz^{\flat}\right)
  -\rhoz\, \mathbf{g}\cdot \uuz=0.
\end{equation*}
Using~\eqref{eq:LinearMomentumCCM}, we furthermore have (omitting the bars, with $\nabla$ still the covariant derivative for the Euclidean metric $q$),
\begin{equation*}
  \dive \big(\bsigma\cdot\uu^{\flat} \big)
  =\uu \cdot \dive\bsigma + \bsigma:\bd =\frac{1}{2} \rho \left( \frac{\partial \uu^{2}}{\partial t}+ \nabla_{\uu} \uu^{2} \right)+ \bsigma:\bd - \rho \, \mathbf{g}\cdot \uu ,
\end{equation*}
since $\bsigma$ is symmetric, and where
\begin{equation*}
  \bd:= \frac{1}{2}\left(\nabla \uu^{\flat} +(\nabla \uu^{\flat})^{\star} \right),
\end{equation*}
is the classical strain rate tensor. We end up with the standard expression of \emph{internal energy balance} in Classical Continuum Mechanics \cite{MH1994,LC1985},
\begin{equation*}
  \rho  \left( \frac{\partial e}{\partial t}+ \nabla_{\uu} e \right)= \bsigma:\bd ,
\end{equation*}
with no heat transfer, and where $e:=E/\rho$ is the specific internal energy.

% ----------------------------------------------------------------
\section{Conclusion}
% ----------------------------------------------------------------

We have revisited Souriau's variational formulation of Relativistic Hyperelasticity. This theory was derived in 1958 with the mindset of Gauge theory: the \emph{perfect matter field} $\Psi$ is somehow similar to the wave function $\psi$ in Quantum Mechanics, but at a macroscopic scale and without the same interpretation. The role of the three-dimensional \emph{body} $\body$, which labels the material points constitutive of the continuous medium under study in the Universe, has been emphasized: it is common to the Relativistic Hyperelasticity theory and to the three-dimensional Classical Continuum Mechanics theory. The body $\body$ is naturally distinguished from a reference configuration $\Omega_{t_{0}}$ in the present Relativistic framework, since $\body$ is embedded into the (non-physical) vector space $V$, whereas $\Omega_{t_{0}}$ is a spacelike submanifold of the Universe $\mM$. In both the Classical and Relativistic frameworks, the body $\body$ is endowed with a volume form, the \emph{mass measure} $\mu$, and a \emph{fixed Riemannian metric} $\bgamma_{0}$. Since this is shared by both theories, we have tried to make a parallel, when possible, and to point out the differences. Our point of view is mainly oriented towards mechanics, rather than astrophysics.

The fundamental observation of Souriau is that the formulation of general covariant constitutive equations for perfect matter involve the metric $g$ only through the \emph{conformation}, defined by
\begin{equation*}
  \bH=(T\Psi) g^{-1} (T\Psi)^{\star}.
\end{equation*}
It is a non degenerate contravariant three-dimensional tensor valued function which plays the role of strain, or more precisely of the inverse of the right Cauchy--Green tensor $\bC$. The connections between $\bH$, the four-dimensional degenerate metric $h=g+\bU^{\flat}\otimes \bU^{\flat}$ and the four-dimensional degenerate co-metric $h^{\sharp}=g^{-1}+\bU\otimes \bU$ (considered by Carter and Quintana \cite{CQ1972}, for instance) are given by lemma~\ref{lem:h-H},
\begin{equation*}
  \bH=(T\Psi) h^{\sharp} (T\Psi)^{\star},
  \qquad
  h=(T\Psi)^{\star} \bH^{-1}T\Psi.
\end{equation*}
Thanks to these formulas, all the definitions of a strain tensor can be expressed using a comparison between the inverse of the conformation $\bH^{-1}$ and $\bH_{0}^{-1}:=\bgamma_{0} \circ \Psi $, where $\bgamma_{0}$ is a reference metric on the body $\body$. Among these definitions, we mention
\begin{equation*}
  \mathfrak E:= \frac{1}{2} \left(\bH^{-1}-\bH_{0}^{-1}\right)
  \quad \text{and} \quad
  \widehat{\mathfrak E} := -\frac{1}{2}\log \big(\bH\, \bH_{0}^{-1}\big).
\end{equation*}
The links between the different metrics and strain tensors encountered in the literature, either defined on the World tube $\mW$, or on the body $\body$, have been clarified (in \autoref{sec:Conformation-strains}, \autoref{sec:3D-metrics} and \autoref{sec:strains}).

In the framework of Variational General Relativity, the stress-energy tensor $\bT$ derives from a \emph{general covariant Lagrangian} (theorem~\ref{thm:relativistic-hyperelasticity} and remark~\ref{rem:T}) and its decompositions allow for the rigorous formulation of \emph{stress tensors},
\begin{itemize}
  \item first, four-dimensional, such as $\bS$ or $\bSigma$ (with a preference for Eckart--Bennoun definition~\eqref{eq:Hyperelasticity-law-Sigma}),
  \item and, then, three-dimensional, such as the \emph{generalized second Piola--Kirchhoff stress tensor} $\bs$, which is always defined by~\eqref{eq:def-s-3D}, and the \emph{generalized Cauchy stress tensor} $\bsigma$, which definition requires the introduction of a spacetime (with a preference for our definition~\eqref{eq:Hyperelasticity-b}).
\end{itemize}
The full Relativistic Hyperelasticity theory is four-dimensional, but its constitutive laws are essentially three-dimensional and very similar to the ones of Classical Continuum Mechanics, a feature which has been used in~\cite{EBT2006} and~\cite{Bro2021}. We have formalized it in \autoref{sec:hyperelasticity-spacetime} and in \autoref{sec:stress-on-the-body}.

By considering the Schwarzschild spacetime (instead of the flat Minkowski spacetime like Souriau did), we have been able to take into account gravity. Following, this time, Künzle~\cite{Kuen1976}, we have recovered the Newton-Cartan formulation of Hyperelasticity in Galilean Relativity, as the limit $c\to \infty$ of our  relativistic formulation in Schwarzschild spacetime.

% ----------------------------------------------------------------
\appendix
% ----------------------------------------------------------------

% ----------------------------------------------------------------
\section{Orthogonal decomposition of four-dimensional 2nd-order tensors}
\label{sec:Orth-Decomp-2nd}
% ----------------------------------------------------------------

We detail in this Appendix the orthogonal decomposition of second-order tensor fields relative to a unit timelike quadrivector $\bW$. This means that we split these tensor fields according to the orthogonal decomposition
\begin{equation*}
  T_{m}\mM = \langle \bW(m) \rangle \oplus \bW(m)^{\perp},
\end{equation*}
where $\bW(m)^{\perp}$ is the orthogonal complement of the one-dimensional timelike subspace $\langle \bW(m) \rangle$, and thus necessarily spacelike.

\begin{itemize}
  \item \emph{For a symmetric covariant second-order tensor field $\bK$}:
        \begin{equation}\label{eq:K-orthogonal-decomposition}
          \bK = \alpha \, \bW^{\flat} \otimes \bW^{\flat} + \bW^{\flat} \otimes \bbeta +\bbeta \otimes \bW^{\flat} + \bk,
        \end{equation}
        where
        \begin{enumerate}
          \item $\alpha : = \bW\cdot \bT\cdot \bW$ is a function,
          \item $\bbeta := -(\bK\cdot \bW + \alpha \bW^{\flat})$ is a covector field orthogonal to $\bW^{\flat}$,
          \item $\bk := \bK - \alpha \, \bW^{\flat} \otimes \bW^{\flat} - \bW^{\flat} \otimes \aaa- \aaa \otimes \bW^{\flat}$ satisfies $\bk \cdot \bW = 0$ and is the spatial part of $\bK$.
        \end{enumerate}

  \item \emph{For a symmetric contravariant second-order tensor field $\bT$}:
        \begin{equation}\label{eq:T-orthogonal-decomposition}
          \bT = \alpha \, \bW \otimes \bW + \bW \otimes \aaa +\aaa \otimes \bW + \bt,
        \end{equation}
        where
        \begin{enumerate}
          \item $\alpha : = \bW^{\flat}\cdot \bT\cdot \bW^{\flat}$ is a function,
          \item $\aaa := -(\bT\cdot \bW^{\flat} + \alpha \bW)$ is a vector field orthogonal to $\bN$,
          \item $\bt := \bT - \alpha \, \bW \otimes \bW - \bW \otimes \aaa- \aaa \otimes \bW$ satisfies $\bt \cdot \bW^{\flat} = 0$ and is the spatial part of $\bT$.
        \end{enumerate}
\end{itemize}

\begin{exam}
  For $\bK=g$, the four-dimensional metric on $\mM$, we get
  \begin{equation*}
    g = k-\bW^\flat \otimes \bW^{\flat},
  \end{equation*}
  where $k$ is determined by $k\cdot \bW = 0$ and $k = g$ on $\bW^{\perp}$. For $\bT = g^{-1}$, the co-metric, we get
  \begin{equation*}
    g^{-1} = k^{\sharp}-\bW \otimes \bW,
  \end{equation*}
  where $k^{\sharp}=g^{-1} k g^{-1}$.
\end{exam}

The following result (see~\cite{Eck1940} or~\cite[page 6]{Lic1955}) is a consequence of \emph{Sylvester's law of inertia}.

\begin{lem}\label{lem:h-signature}
  Let $k$ be the spatial part of the metric $g$ in the orthogonal decomposition relative to a unit timelike quadrivector $\bW$. Then,
  $k$ is positive definite. In particular, the signature of $k$ is~$(0, +, +, +)$.
\end{lem}

% ----------------------------------------------------------------
\section{Schwarzschild spacetime}
\label{sec:Schwarzschild-spacetime}
% ----------------------------------------------------------------

According to Birkhoff's theorem~\cite{BL1923,HE1973}, the only spherically symmetric solution of Einstein's equation \emph{in the vacuum} with vanishing cosmological constant is the exterior \emph{Schwarzschild metric}. It is a \emph{static} metric which describes the gravity field outside from a (spherical, nonrotating) massive planet ---or a star or a blackhole--- of mass $M$~\cite{Sch1916,MTW1973} and is written as
\begin{equation*}
  g = -\left(1-\frac{2GM}{c^{2} r}\right) c^{2} \dd t^{2}
  +\left(1-\frac{2GM}{c^{2} r}\right)^{-1} \dd r^{2}
  +r^{2}\left(\dd\theta^{2} + \sin^{2}\theta\, \dd \varphi^{2}\right),
  \qquad
  r>r_{\mathrm{s}}=\frac{2GM}{c^{2}},
\end{equation*}
where $G$ is the gravitational constant, $\theta$ is the colatitude (angle from North pole), $\varphi$ is the longitude, and $r_{\mathrm{s}}$ is the Schwarzschild radius. The surface of the planet (or star) is at radius $r=r_{0}$ much larger than $r_{\mathrm{s}}$. The  coordinate transformation,
\begin{equation*}
  \bar{r}=\frac{1}{4} \left(2r -r_{\mathrm{s}} + 2 \sqrt{r(r-r_{\mathrm{s}})}\right),
\end{equation*}
allows first to express the Schwarzschild metric into the so-called \emph{isotropic coordinates} expression~\cite[p. 840]{MTW1973}
\begin{equation*}
  g=-\left(\frac{1-{\bar{r}_{\mathrm{s}}}/{\bar{r}}}{1+{\bar{r}_{\mathrm{s}}}/{\bar{r}}}\right)^{2} c^{2} \dd t^{2}
  +\left(1+\frac{\bar{r}_{\mathrm{s}}}{\bar{r}}\right)^{4}\left[ \dd \bar{r}^{2}
    +\bar{r}^{2}\left(\dd\theta^{2} + \sin^{2}\theta\, \dd \varphi^{2}\right)\right],
  \qquad
  \bar{r}>\bar{r}_{\mathrm{s}}=\frac{r_{\mathrm{s}}}{4},
\end{equation*}
and, then, to put it in the \emph{Cartesian isotropic coordinates} expression (with $\bar{r}:=\sqrt{\delta_{ij} x^{i} x^{j}}$, null at the center of the planet/star),
\begin{equation*}
  g=-\left(\frac{1-{\bar{r}_{\mathrm{s}}}/{\bar{r}}}{1+{\bar{r}_{\mathrm{s}}}/{\bar{r}}}\right)^{2}  \dd {x^{0}}^{2}
  +\left(1+\frac{\bar{r}_{\mathrm{s}}}{\bar{r}}\right)^{4}\left[ \dd {x^{1}}^{2}+  \dd {x^{2}}^{2}+ \dd {x^{3}}^{2}\right],
  \qquad
  \bar{r}>\bar{r}_{\mathrm{s}}=\frac{r_{\mathrm{s}}}{4},
\end{equation*}
where we have set $x^{0}=ct$.

% ----------------------------------------------------------------
\section{Divergences in a static spacetime}
\label{sec:divergences-spacetime}
% ----------------------------------------------------------------

For an arbitrary metric $g$ and in an arbitrary coordinate system $(x^{\mu})$, the divergence of a quadrivector $\bP$ and of a second-order contravariant tensor $\bT$ are given by
\begin{align}
  \label{eq:divX}
  \dive^{g} \bP         & = \partial_{\mu} P^{\mu}+ \Gamma^{\nu}_{\nu \mu } P^{\mu},
  \\
  \label{eq:divT}
  (\dive^{g} \bT)^{\mu} & = \partial_{\nu} T^{\mu \nu}
  + \Gamma^{\rho}_{\rho \nu}  T^{\mu\nu}+ \Gamma^{\mu}_{\nu\rho} T^{\rho\nu},
\end{align}
where $\Gamma^{\lambda}_{\mu\nu}$ are the Christoffel symbols of the metric $g$, given by the standard formula
\begin{equation}\label{eq:Gammas}
  \Gamma^{\lambda}_{\mu \nu} = \frac{1}{2} g^{\lambda\sigma} \left( \frac{\partial g_{\sigma \mu}}{\partial x^{\nu}} + \frac{\partial g_{\sigma \nu}}{\partial x^{\mu}} - \frac{\partial g_{\mu \nu}}{\partial x^{\sigma}} \right).
\end{equation}

Suppose now that $g$ is a static metric and that the coordinate system $(x^{\mu})$ is chosen such that
\begin{equation*}
  \frac{\partial g_{\mu \nu}}{\partial x^{0}} = 0 \quad \text{and} \quad g_{0i}=0,
\end{equation*}
meaning that the Universe metric $g$ does not depend on $x^{0}$ and that it is related to the spatial metric $\mg$ by
\begin{equation*}
  g = g_{00} (\dd x^{0})^{2} + \mg, \quad \text{where} \quad \mg = g_{ij} \dd x^{i} \dd x^{j}.
\end{equation*}
Then,
\begin{enumerate}
  \item $\Gamma^{0}_{00}=\Gamma^{0}_{ij}=\Gamma^{i}_{j0}=0$,
  \item the Christoffel symbols $\bar \Gamma^{i}_{jk}$ of the 3D spatial metric $\mg=(g_{ij})$ are equal to the spatial Christoffel symbols $\Gamma^{i}_{jk}$ of the 4D metric $g$,
        \begin{equation*}
          \bar\Gamma^{i}_{jk}=\Gamma^{i}_{jk},
        \end{equation*}
  \item and, moreover
        \begin{equation*}
          \Gamma^{i}_{00}=(\grad^{\mg}\!\! \sqrt{- g_{00}} )^{i},
          \quad\text{and}\quad
          \Gamma^{0}_{0 i}=(\dd \ln\sqrt{- g_{00}} )_{i},
        \end{equation*}
\end{enumerate}
where $\grad^{\mg}\!\! f:=(\mg)^{\sharp} \dd f$, when $f$ is independent of $x^{0}$.

We get therefore
\begin{equation}\label{eq:divP-static}
  \begin{split}
    \dive^{g} \bP & = \frac{\partial P^{0}}{\partial x^{0}} + \frac{\partial P^{i}}{\partial x^{i}}+ \left(\Gamma^{j}_{j i}+\Gamma^{0}_{0 i }  \right) P^{i}
    \\
    & = \frac{\partial P^{0}}{\partial x^{0}} +\dive^{\mg} \bP^{\top}+ \bP^{\top}\cdot \dd \ln\sqrt{- g_{00}},
  \end{split}
\end{equation}
where we have set $\bP^{\top} := P^{i}\partial x^{i}$.

Setting now, in the coordinate system $(x^{\mu})$,
\begin{equation*}
  \bT =
  \begin{pmatrix}
    T^{00} & \aaa^{\star}
    \\
    \aaa   & \bt
  \end{pmatrix}
  ,
\end{equation*}
where $\aaa:=T^{0i}\partial_{x^{i}}$ and $\bt:=T^{ij}\partial_{x^{i}}\partial_{x^{j}}$, we have
\begin{equation}\label{eq:divT0-static}
  \begin{split}
    (\dive^{g} \bT)^{0} & = \frac{\partial T^{0 0}}{\partial x^{0}}+ \frac{\partial a^{i}}{\partial x^{i}}
    + (\Gamma^{j}_{j i}+3\Gamma^{0}_{0 i})a^{i}
    \\
    & = \frac{\partial T^{0 0}}{\partial x^{0}}+\dive^{\mg} \aaa + 3\, \aaa \cdot  \dd \ln\sqrt{- g_{00}},
  \end{split}
\end{equation}
and
\begin{equation}\label{eq:divTi-static-components}
  (\dive^{g} \bT)^{i}  =\frac{\partial a^{i}}{\partial x^{0}}+\frac{\partial \mathrm t^{i j}}{\partial x^{j}}
  + (\Gamma^{k}_{k j}+\Gamma^{0}_{0 j})  \mathrm t^{ij}+ \Gamma^{i}_{jk} \mathrm t^{jk}+\Gamma^{i}_{00} T^{00},
\end{equation}
this last equation being recast more intrinsically as
\begin{equation}\label{eq:divTi-static}
  (\dive^{g} \bT)^{\top} =\frac{\partial \aaa}{\partial x^{0}}+ \dive^{\mg} \bt
  +  \bt \cdot \dd \ln\sqrt{- g_{00}} +T^{00} \grad^{\mg} \!\!\sqrt{- g_{00}}.
\end{equation}

\begin{rem}
  In~\cite{Yor1979} and more recently in~\cite[Chapter 4]{Gou2012}, equations~\eqref{eq:divX} and~\eqref{eq:divT} are expressed in an intrinsic manner using the so-called $(3+1)$-orthogonal decomposition of the divergence operator obtained through the theory of (pseudo-)Riemannian hypersurfaces~\cite[Chapter 5]{GHL2004}, and which is similar to the one used in Thick Shell Theory.
\end{rem}

% ----------------------------------------------------------------
\section{Three-dimensional Riemannian metrics and mass densities}
\label{sec:3D-metrics}
% ----------------------------------------------------------------

We assume in this Appendix that a time function $\hat{t}$ is given, inducing a spacetime structure on $\mW$ and we denote by $\Omega_{t}$ the corresponding spacelike hypersurfaces.

\subsection*{3D Riemannian metrics on the hypersurfaces $\Omega_{t}$}

Each three-dimensional manifold $\Omega_{t}$ is endowed with two Riemannian metrics. The first one is just the restriction $j_{t}^{*}g$ of the four-dimensional Universe metric $g$ and coincides with $j_{t}^{*} \mg$ (since $\mg$ is the spatial component of $g$ in its orthogonal decomposition~\eqref{eq:g3D} relative to $\bN$, the unit normal to $\Omega_{t}$). The second one is the restriction $j_{t}^{*} h$ of the degenerate metric $h=g+\bU^{\flat}\otimes \bU^{\flat}$ (the spatial part of $g$ in its orthogonal decomposition~\eqref{eq:def-h} relative to $\bU$). Note that, unless $\Omega_{t}$ is orthogonal to $\bU$, these two metrics on $\Omega_{t}$ do not match. However, the following lemma allows to relate their respective Riemannian volume forms $\vol_{j_{t}^{*}g}$ and $\vol_{j_{t}^{*}h}$ on~$\Omega_{t}$.

\begin{lem}\label{lem:voljth-voljtg}
  We have
  \begin{equation}\label{eq:voljth-voljtg}
    \vol_{j_{t}^{*} h} = \gamma \, \vol_{j_{t}^{*} g} = j_{t}^{*} \left(i_{\bU}\vol_{g}\right),
    \quad
    \text{on $\Omega_{t}$.}
  \end{equation}
  where $\gamma=-\langle \bU, \bN\rangle_{g}$ is the generalized Lorentz factor.
\end{lem}

\begin{proof}
  We have first
  \begin{equation*}
    \vol_{j_{t}^{*} h} = \sqrt{\det[(j_{t}^{*} g)^{-1} j_{t}^{*} h]}\, \vol_{j_{t}^{*} g}.
  \end{equation*}
  Let $\xx \in \Omega_{t}$ and let $(\ee_{i})$ be an orthonormal basis of $T_{\xx}\Omega_{t}$ for the metric $j_{t}^{*} g$. Then,
  \begin{equation*}
    (\ee_{0} = \bN, T_{\xx}j_{t}.\ee_{i})
  \end{equation*}
  is an orthonormal basis of $T_{m}\mM$ for the Lorentzian metric $g$, and we will denote by $U^{\mu}$, the components of $\bU(m)$ in this basis. Now, using~\eqref{eq:def-h}, and the fact that $T_{\xx}j_{t}.\ee_{0}= 0$, $T_{\xx}j_{t}.\ee_{i}=\ee_{i}$, we get
  \begin{equation*}
    [(j_{t}^{*} g)^{-1}]^{ij} = \delta^{ij}, \qquad [j_{t}^{*} h]_{ij} = \delta_{ij} + \delta_{ik}U^{k}\delta_{il}U^{l}.
  \end{equation*}
  Hence, we are reduced to calculate the determinant
  \begin{equation*}
    \det[(j_{t}^{*} g)^{-1} j_{t}^{*} h] = \det(I_{3} + UU^{\star}),
  \end{equation*}
  where
  \begin{equation*}
    U =
    \begin{pmatrix}
      U^{1} \\
      U^{2} \\
      U^{3}
    \end{pmatrix},
    \quad \text{and} \quad U^{\star} =
    \begin{pmatrix}
      U^{1} & U^{2} & U^{3}
    \end{pmatrix}.
  \end{equation*}
  Now the $3 \times 3$ matrix $UU^{\star}$ has a double eigenvalue $0$ and a single eigenvalue $U^{\star}U$ and thus
  \begin{equation*}
    \det\left(I_{3} + UU^{\star}\right) = 1 + U^{\star}U = 1 + \sum (U^{i})^{2} = (U^{0})^{2},
  \end{equation*}
  since
  \begin{equation*}
    \norm{\bU}_{g}^{2} = -(U^{0})^{2} +  \sum (U^{i})^{2} = -1.
  \end{equation*}
  Therefore, we get
  \begin{equation*}
    \sqrt{\det[(j_{t}^{*} g)^{-1} j_{t}^{*} h]} = \sqrt{(U^{0})^{2}} = -\langle \bU, \bN\rangle_{g},
  \end{equation*}
  because
  \begin{equation*}
    - \langle \bU, \bN\rangle_{g} = U^{0},
  \end{equation*}
  and $\langle \bU, \bN\rangle_{g}$ is assumed to be negative. Now, we have
  \begin{equation*}
    \vol_{j_{t}^{*} g} = j_{t}^{*}(i_{\bN} \vol_{g}),
  \end{equation*}
  and thus
  \begin{equation*}
    \vol_{j_{t}^{*} h} = \gamma \, \vol_{j_{t}^{*} g} = \gamma \,j_{t}^{*}(i_{\bN} \vol_{g}) = j_{t}^{*}(i_{\bU} \vol_{g}).
  \end{equation*}
\end{proof}

\subsection*{Geometric interpretation of the relativistic mass density $\rho$}

By multiplying~\eqref{eq:voljth-voljtg} by the rest mass density $\rho_{r}$ and using the definition $i_{\bP} \vol_{g}=\rho_{r} i_{\bU} \vol_{g}=\Psi^{*}\mu$, where $\mu$ is the mass measure on the body $\body$, we get the following equalities on $\Omega_{t}$,
\begin{equation*}
  \Psi_{t}^{*}\mu = j_{t}^{*} \left(\rho_{r} i_{\bU}\vol_{g}\right) = \rho_{r} j_{t}^{*} \left(i_{\bU}\vol_{g}\right)
  = \rho_{r}\gamma \,\vol_{j_{t}^{*} g},
\end{equation*}
summarized as
\begin{equation*}
  \Psi_{t}^{*}\mu = \rho_{r} \gamma \,\vol_{j_{t}^{*} g} = \rho \,\vol_{j_{t}^{*} g}.
\end{equation*}
The function~\eqref{eq:rho},
\begin{equation*}
  \rho:= \gamma \rho_{r},
\end{equation*}
defined on the World tube $\mW$, is interpreted as the \emph{relativistic mass density}, \emph{i.e.}, the mass density measured on $\Omega_{t}$, relatively to the 3D metric $j_{t}^{*} g$.

\subsection*{3D Riemannian metrics and mass densities on the body $\body$}

\emph{If we make the stronger assumption that $\Psi_{t} \colon \Omega_{t} \to \body$ is a diffeomorphism,} then, the conformation induces a one-parameter family $\bgamma(t)$ of three-dimensional Riemannian metrics on the three-dimensional body $\body$
\begin{equation}\label{eq:def-gamma-t}
  \bgamma(t)^{-1} := \bH \circ j_{t} \circ \Psi_{t}^{-1}, \qquad \bgamma(t) := \bH^{-1} \circ j_{t} \circ \Psi_{t}^{-1}.
\end{equation}
The metric $\bgamma(t)$ is the true analogue of the right Cauchy--Green tensor $\bC:=\bF^{\star} q\, \bF$. Indeed, we have the identification $\bgamma\equiv \bC$ in Classical Continuum Mechanics when the body $\body$ is identified with a reference configuration $\Omega_{0}$~\cite{Nol1978,Rou2006,KD2021}. Note however that in the non-relativistic case, the metric $\bgamma$ on $\body$ is the pull-back $\pp^{*}q$ of the Euclidean metric $q$ on the space $\espace$ by the embedding $\pp\colon \body \to \espace$, whereas in~\eqref{eq:def-gamma-t}, it is defined using the conformation and a foliation of the World tube $\mW$. The following result relates $\bgamma(t)$ with the degenerate quadratic form $h$ defined by~\eqref{eq:def-h}.

\begin{lem}\label{lem:gamma-h}
  On $\Omega_{t}$, we have
  \begin{equation*}
    \Psi_{t}^{*} \bgamma(t) = j_{t}^{*}h,
  \end{equation*}
  where $h=g+\bU^\flat \otimes \bU^{\flat}$ and $ j_{t}^{*}h=(Tj_{t})^{\star} (h\circ j_{t}) Tj_{t}$.
\end{lem}

\begin{proof}
  We have $\bgamma(t) := \bH^{-1} \circ j_{t} \circ  \Psi_{t}^{-1}$, and $ h=(T\Psi)^{\star} \bH^{-1} T\Psi$ by lemma~\ref{lem:h-H}. Therefore
  \begin{align*}
    \Psi_{t}^{*} \bgamma(t) & = (T\Psi_{t})^{\star} (\bgamma(t) \circ \Psi_{t}) T\Psi_{t}            \\
                            & = (T\Psi_{t})^{\star} (\bH^{-1} \circ j_{t}) T\Psi_{t}                 \\
                            & = (Tj_{t})^{\star}(T\Psi)^{\star} (\bH^{-1}\circ j_{t}) T\Psi \,Tj_{t} \\
                            & = (Tj_{t})^{\star} (h \circ j_{t})Tj_{t}                               \\
                            & = j_{t}^{*}h.
  \end{align*}
\end{proof}

To the three-dimensional Riemannian metric $\bgamma(t)$ on $\body$ is associated a three-dimensional volume form $\vol_{\bgamma(t)}$. Since the body $\body$ is initially endowed with a mass measure $\mu$ and a fixed metric $\bgamma_{0}$ (see \autoref{sec:matter-fields}),
mass conservation can then be expressed on the body exactly as in the intrinsic Lagrangian formulation of Classical Continuum Mechanics \cite{KD2021}, \emph{i.e.}, as
\begin{equation*}
  \mu = \rho_{\bgamma(t)} \vol_{\bgamma(t)} = \rho_{\bgamma_{0}} \vol_{\bgamma_{0}},
\end{equation*}
where  $\rho_{\bgamma(t)}$ and  $\rho_{\bgamma_{0}}$ are mass densities on $\body$. In the following lemma, we relate $\rho_{\bgamma(t)}$ with the rest mass density $\rho_{r}$, defined by~\eqref{eq:def-rho}.

\begin{lem}\label{lem:mass-densities}
  Let $\rho_{\bgamma(t)}$ be the mass density on the body $\body$ defined implicitly by $\mu = \rho_{\bgamma(t)} \vol_{\bgamma(t)}$.
  Then, we have,
  \begin{equation*}
    \Psi_{t}^{*} \rho_{\bgamma(t)} = j_{t}^{*}\rho_{r}.
  \end{equation*}
\end{lem}

\begin{proof}
  We have $\Psi^{*}\mu = \rho_{r} i_{\bU} \vol_{g}$ by~\eqref{eq:def-P}--\eqref{eq:def-U}, and hence
  \begin{equation*}
    \Psi_{t}^{*}(\rho_{\bgamma(t)} \vol_{\bgamma(t)}) = j_{t}^{*}\Psi^{*}(\rho_{\bgamma(t)} \vol_{\bgamma(t)}) = j_{t}^{*}\Psi^{*}\mu = j_{t}^{*}(\rho_{r} i_{\bU} \vol_{g}).
  \end{equation*}
  We get therefore
  \begin{equation*}
    (\Psi_{t}^{*}\rho_{\bgamma(t)}) \vol_{\Psi_{t}^{*}\bgamma(t)} = (j_{t}^{*}\rho_{r}) j_{t}^{*}(i_{\bU} \vol_{g}),
  \end{equation*}
  but $ j_{t}^{*}(i_{\bU} \vol_{g}) = \vol_{j_{t}^{*}h}$ by lemma~\ref{lem:voljth-voljtg} and  $\Psi_{t}^{*}\bgamma(t) = j_{t}^{*}h$ by lemma~\ref{lem:gamma-h}. We get thus
  \begin{equation*}
    (\Psi_{t}^{*}\rho_{\bgamma(t)}) \vol_{j_{t}^{*}h} = (j_{t}^{*}\rho_{r})  \vol_{j_{t}^{*}h},
  \end{equation*}
  which ends the proof.
\end{proof}

% ----------------------------------------------------------------
\section{Choice of a reference metric}
\label{sec:reference-metric}
% ----------------------------------------------------------------

\subsection*{Reference metric on the body $\body$}

There are several choices for a \emph{reference metric} on the body $\body$. One possibility is to endow the body with an arbitrary fixed metric $\bgamma_{0}$ (for example $\bgamma_{0}=q$, the Euclidean metric, in \cite{Sou1958,Sou1964}). But when a spacetime and the associated spacelike hypersurfaces $\Omega_{t}$ are introduced, with in particular the choice of a reference configuration $\Omega_{t_{0}}$, and \emph{when the restriction $\Psi_{t_{0}}=j_{t_{0}}^{*} \Psi$ of the matter field to $\Omega_{t_{0}}$ is a diffeomorphism}, then two other ---mechanistic--- possibilities are offered:
\begin{itemize}
  \item[(a)] either to consider as reference metric on the body $\body$, the Riemannian metric $\bgamma(t=t_{0})$ at initial time $t_{0}$,
    \begin{equation*}
      \bgamma_{0}^{a} := \bgamma(t_{0})=(\Psi_{t_{0}})_{ *} j_{t_{0}}^{*}h,
    \end{equation*}
    where the second equality is due to lemma~\ref{lem:gamma-h},

  \item[(b)] or to endow the body $\body$ with the Riemannian metric
    \begin{equation*}
      \bgamma_{0}^{b}: =(\Psi_{t_{0}})_{ *} j_{t_{0}}^{*}g ,
    \end{equation*}
    obtained as the pushforward on the body, of the restriction $j_{t_{0}}^{*}g$ of the Universe metric to $\Omega_{t_{0}}$.
\end{itemize}
These two reference metrics do not coincide in general. In case (a), the mixed tensor $(\bgamma_{0}^{a})^{-1} \bgamma(t)$ is equal to the identity at $t = t_{0}$. In case (b), which mimics what is done in non relativistic three-dimensional Hyperelasticity \cite{Rou2006,KD2021}, $(\bgamma_{0}^{b})^{-1} \bgamma(t_{0}) = (\bgamma_{0}^{b})^{-1} \bgamma_{0}^{a}\neq \id$ in general.

The question of which reference metric is to be prefered is in fact related to the difficult question of the definition of an associated reference stress-free state (at which $\bsigma=0$). This question arise naturally when one choose an explicit expression for the specific internal energy $e$ (such as Money--Rivlin's \cite{Moo1940}, Hart--Smith's \cite{Har1966}, Ogden's \cite{Ogd1972}, Arruda--Boyce's \cite{AB1993} or others \cite{GMDC2011}).
Fortunately for Mechanics, the difference between $ \bgamma_{0}^{b}$ and $ \bgamma_{0}^{a}$ is only due to relativistic effects, since by~\eqref{eq:U-orthogonal-decomposition} the restriction $ j_{t_{0}}^{*} \left(\bU^{\flat}\otimes \bU^{\flat}\right)$ is in $1/c^{2}$.

\subsection*{Frozen metric on the World tube $\mW$}

As mentioned in \autoref{sec:Conformation-strains}, instead of explicitly introducing a reference metric $\bgamma_{0}$ on the body, some authors consider a reference degenerate quadratic form $h_{0}$ of signature $(0,+,+,+)$ on the World tube \cite{KM1992,KM1997}, with some additional properties, leading them to call it a \emph{frozen metric}. Such a four-dimensional frozen metric is in fact strongly related to a three-dimensional reference metric on the body $\body$. The following result provides necessary conditions for a given quadratic form $h_{0}$ on $\mW$ to be the pullback of a fixed Riemannian metric $\bgamma_{0}$ on $\body$ by the matter field $\Psi$.

\begin{lem}[Kijowski and Magli, 1997]\label{lem:KM}
  Let $h_{0}$ be a field of quadratic forms on the World tube $\mW$. Then, necessary conditions for the existence of a Riemannian metric $\bgamma_{0}$ on $\body$ such that $h_{0} = \Psi^{*}\bgamma_{0}$ are
  \begin{equation*}
    h_{0}\bU=0, \quad \text{and} \quad \Lie_{\bU} h_{0}=0.
  \end{equation*}
  Such a quadratic form is necessarily of signature $(0,+,+,+)$.
\end{lem}

\begin{proof}
  Suppose that $h_{0} = \Psi^{*}\bgamma_{0}$. Since $T\Psi. \bU=0$, we get first that
  \begin{equation*}
    h_{0}\bU = (T\Psi)^{\star} (\bgamma_{0}\circ \Psi) T\Psi. \bU = 0,
  \end{equation*}
  and that $h_{0}$ is of signature $(0,+,+,+)$, since $\Psi$ is assumed to be a submersion on $\mW$. Now, let $\varphi^{t}$ be the flow of the vector field $\bU$. Then, we have
  \begin{equation*}
    \partial_{t} \left[(\Psi\circ\varphi^{t})(m)\right] = T_{\varphi^{t}(m)}\Psi . \bU(\varphi^{t}(m)) = 0, \qquad \forall t, \quad \forall m \in \mW,
  \end{equation*}
  and thus $\Psi\circ\varphi^{t} = \Psi\circ\varphi^{0} = \Psi$. Hence, we get
  \begin{equation*}
    (\varphi^{t})^{*} h_{0} = (\varphi^{t})^{*} \Psi^{*}\bgamma_{0} = (\Psi \circ \varphi^{t})^{*}\bgamma_{0} = \Psi^{*}\bgamma_{0} = h_{0}
  \end{equation*}
  and $\Lie_{\bU} h_{0} = \left[\partial_{t}(\varphi^{t})^{*} h_{0}\right]_{t=0}= 0$.
\end{proof}

% ----------------------------------------------------------------
\section{Three-dimensional strains}
\label{sec:strains}
% ----------------------------------------------------------------

When the World tube $\mW$ is foliated by spacelike hypersurfaces $\Omega_{t}$ and \emph{when the restriction $\Psi_{t}$ of the matter field to $\Omega_{t}$ is a diffeomorphism}, any of the three following 3D symmetric covariant tensor fields
\begin{equation*}
  \bH^{-1} \circ j_{t} \; \text{($\Sym^{2}V$-vector valued, on $\Omega_{t}$)}, \qquad \bgamma(t)\; \text{(on $\body$)}, \qquad \text{and} \qquad j_{t}^{*}h \; \text{(on $\Omega_{t}$)},
\end{equation*}
leads to equivalent formulations of Relativistic Hyperelasticity models. Indeed, these tensor fields are related to each other by
\begin{align*}
  \text{on $\body$}: \qquad\quad                                    & \bgamma(t)=\bH^{-1} \circ j_{t} \circ \Psi_{t}^{-1}=\Psi_{t*} \, j_{t}^{*}h,
  \\
  \text{$\Sym^{2}V$-vector valued, on $\Omega_{t}$:} \, \qquad\quad & \bH^{-1} \circ j_{t}=(\Psi_{t*} \, j_{t}^{*}h)\circ \Psi_{t}=\bgamma(t)\circ \Psi_{t},
  \\
  \text{on $\Omega_{t}$:}\qquad\quad                                & j_{t}^{*}h =(\Psi_{t})^{*}\bgamma(t)=(\Psi_{t})^{*}(\bH^{-1} \circ j_{t} \circ \Psi_{t}^{-1}),
\end{align*}
Making use of~\eqref{eq:strains-H}, the associated ---\emph{in fine} equivalent--- definitions of strain tensors are then the following
\begin{align*}
  \text{on $\body$}: \qquad\quad                                    & \frac{1}{2}\left(\bgamma(t)-\bgamma_{0}\right)=\mathfrak E \circ \Psi_{t}^{-1},
                                                                    &                                                                                  & \frac{1}{2}\log \big(\bgamma_{0}^{-1} \bgamma(t)\big)=\widehat{\mathfrak E} \circ \Psi_{t}^{-1},
  \\
  \text{$\Sym^{2}V$-vector valued, on $\Omega_{t}$:} \, \qquad\quad & \frac{1}{2}\left(\bH^{-1}-\bH_{0}^{-1}\right) \circ j_{t}=j_{t}^{*} \mathfrak E,
                                                                    &                                                                                  & -\frac{1}{2}\log \big( \bH\, \bH_{0}^{-1}\big)  \circ j_{t} =j_{t}^{*} \widehat{\mathfrak E},
  \\
  \text{on $\Omega_{t}$:}\qquad\quad                                & \frac{1}{2}\left( j_{t}^{*}h- j_{t}^{*}h_{0}\right)= j_{t}^{*} \be ,
                                                                    &                                                                                  & \frac{1}{2}\log \big((j_{t}^{*}h_{0})^{-1} j_{t}^{*}h\big),
\end{align*}
where $h_{0}=\Psi^{*} \bgamma_{0}$ is the so-called frozen metric on the World tube $\mW$,  $j_{t}^{*} h_{0}=\Psi_{t}^{*} \bgamma_{0}$ is its restriction to $\Omega_{t}$,
and $\bH_{0}= \bgamma_{0}^{-1}\circ \Psi$. Note that $t=t_{0}$ can be set in the above restrictions to obtain definitions on $\Omega_{t_{0}}$.

% ----------------------------------------------------------------
\section{Three-dimensional stresses}
\label{sec:stress-on-the-body}
% ----------------------------------------------------------------

Given a spacetime structure on the body World tube $\mW$ and the corresponding orthogonal decomposition relative to $\bN$, the normal to the hypersurfaces $\Omega_{t}$, the generalized Cauchy stress $\bsigma$, defined here as the spatial part of the four-dimensional stress~$\bSigma$ (remark~\ref{rem:T}), has for expression~\eqref{eq:Hyperelasticity-b-bis},
\begin{equation*}
  \bsigma(m) = \rho_{r}(m)\, (g_{m}^{3D})^{\sharp} (T_{m}\Psi)^{\star} \, \bs(m)\, (T_{m}\Psi) (g_{m}^{3D})^{\sharp} ,
  \qquad m\in \mW,
\end{equation*}
where $\mg=g + \bN^{\flat} \otimes \bN^{\flat}$ is the spatial part of $g$ (see~\eqref{eq:g3D}), and $\bs$ is the covariant stress tensor defined by~\eqref{eq:s-3D}. Since, by its very definition, $\bsigma$ has values in $T\Omega_{t} \otimes T\Omega_{t}$ because $T\Omega_{t} = \bN^{\bot}$, the mapping
\begin{equation*}
  \bsigma(j_{t}(\xx)) = \rho_{r}(j_{t}(\xx))\, (g_{j_{t}(\xx)}^{3D})^{\sharp} (T_{j_{t}(\xx)}\Psi)^{\star} \, \bs(j_{t}(\xx))\, (T_{j_{t}(\xx)}\Psi) (g_{j_{t}(\xx)}^{3D})^{\sharp}.
\end{equation*}
is a second-order contravariant tensor field on the three-dimensional manifold $\Omega_{t}$.

In the particular case of the Schwarzschild spacetime described in \autoref{sec:hyperelasticity-Schwarzschild}, where, $q$ denoting the Euclidean metric,
\begin{equation*}
  \mg = kq, \quad \text{and} \quad (\mg)^{\sharp} = k^{-1}q^{-1},
\end{equation*}
the three-dimensional stress $\bsigma\circ j_{t}$ is given by
\begin{equation*}
  \bsigma\circ j_{t} =  \frac{1}{\gamma k^{2}} \rho\, q^{-1}  (T\Psi_{t})^{\star} \left(\bs\circ j_{t}\right)  (T\Psi_{t}) q^{-1},
\end{equation*}
with the abuse of notation $(\rho_{r}/k^{2}) \circ j_{t}= \rho_{r}/k^{2} = \rho/\gamma k^{2}$, and where $\gamma$, not to be confused with the metric $\bgamma(t)$ on the body $\body$, is the generalized Lorentz factor~\eqref{eq:gamma-s}.

Let us now make the stronger assumption that \emph{the restriction $\Psi_{t}=\Psi\circ j_{t}$ of the matter field to $\Omega_{t}$ is a diffeomorphism}, and set $\pp := \Psi_{t}^{-1}$ and $\bF=T\pp=T\Psi_{t}^{-1}$, by analogy with Classical Continuum Mechanics (remark~\ref{rem:TPsit-iso}). Then, the stress $\bsigma $ on $\Omega_{t}$ can be recast as the pullback by $\Psi_{t}$
\begin{equation}\label{eq:sigmaRG}
  \bsigma \circ j_{t} = \frac{1}{\gamma k^{2}} \rho\, q^{-1} \big(\Psi_{t}^{*}\btheta^{\flat}\big) q^{-1}= \frac{1}{\gamma k^{2}} \rho\, q^{-1} \big(p_{*}\btheta^{\flat}\big) q^{-1},
\end{equation}
of a covariant stress tensor $\btheta^{\flat}$, defined on $\body$, and given by
\begin{equation*}
  \btheta^{\flat}:= \bs\circ j_{t} \circ\Psi_{t}^{-1}= \bgamma(t)\btheta \bgamma(t),
  \qquad
  \btheta:= 2\bgamma_{0}^{-1}\frac{\partial  w}{\partial \widehat{\bgamma}} ,
  \qquad
  \frac{\partial  w}{\partial \widehat{\bgamma}}=\frac{\partial  w}{\partial \widehat{\bgamma}}(\bgamma_{0}^{-1}\bgamma(t)).
\end{equation*}
Indeed, by definition, $ \bH \circ j_{t} \circ\Psi_{t}^{-1}=\bgamma(t)$ and $\bH_{0}\circ j_{t}\circ \Psi_{t}^{-1}=\bgamma_{0}\circ \Psi\circ j_{t}\circ \Psi_{t}^{-1}=\bgamma_{0}$.

The contravariant stress tensor $\btheta$, defined on the body $\body$, is then recognized as the Rougée stress tensor introduced in~\cite{Rou1991a,Rou2006,KD2021} (and which coincides with the second Piola-Kirchhoff stress tensor when $\body$ is identified with a reference configuration $\Omega_{0}$). In that case, the constitutive equation
\begin{equation*}
  \btheta:= 2\bgamma_{0}^{-1}\frac{\partial  w}{\partial \widehat{\bgamma}} =\btheta(\bgamma)
\end{equation*}
is the formulation of hyperelasticity on the body $\body$ (see~\cite[Chapter XII]{Rou1997}, \cite[Application 1]{Rou2006} and \cite[Theorem 3.4]{KD2021}).

The prefactor $1/\gamma k^{2}$ in~\eqref{eq:sigmaRG} combines both gravitational effects (through the conformal factor $k$) and relativistic effects (through the generalized Lorentz factor $\gamma$). The 3D stress tensor on $\Omega_{t}$
\begin{equation*}
  \btau\circ j_{t}  := \gamma k^{2}\frac{\bsigma\circ j_{t} }{\rho},
  \quad \text{such as}\quad
  \btau\circ j_{t}  =q^{-1} \bF^{-\star}\big(\btheta^{\flat}\circ \Psi_{t}\big) \bF^{-1} q^{-1},
\end{equation*}
is therefore a second relativistic generalization of the Kirchhoff stress tensor $\bsigma/\rho$ of Classical Continuum Mechanics (see remark~\ref{rem:Kirchhoff}), this time dedicated to the Schwarzschild spacetime. Recall that for the flat Minkowski metric we have $k=1$, and that for the Galilean limit, we have $\gamma = 1$.

% ----------------------------------------------------------------

\end{document}